\documentclass[reqno,12pt]{amsart}

\usepackage{amsmath,amssymb,amsthm,amsfonts,amsxtra}
\usepackage{fullpage}

\begin{document}

\newtheorem{thm}{Theorem}
\newtheorem{cor}{Corollary}
\newtheorem{lem}{Lemma}
\newtheorem{prop}{Proposition}

\def\Ref#1{Ref.~\cite{#1}}

\def\const{\text{const.}}
\def\Rnum{{\mathbb R}}
\def\sgn{{\rm sgn\;}}

\def\Esp{{\mathcal E}}
\def\k{\hat{\mathbf k}}

\def\dens{T}
\def\flux{\Phi}
\def\Div{{\rm Div}\,}
\def\triv{\Psi}
\def\trivflux{\Theta}
\def\chargeflux{\Gamma}

\def\V{V}
\def\A{A}

\def\gaugevar{\underline{\boldsymbol\chi}}

\tolerance=10000
\allowdisplaybreaks[4]

\title{Topological charges and conservation laws involving an arbitrary function of time for dynamical PDEs}

\author{
Stephen C. Anco$^1$
\lowercase{\scshape{and}}
Elena Recio$^2$ 
\\\\
${}^1$D\lowercase{\scshape{epartment}} \lowercase{\scshape{of}} M\lowercase{\scshape{athematics and}} S\lowercase{\scshape{tatistics}}\\
B\lowercase{\scshape{rock}} U\lowercase{\scshape{niversity}}\\
S\lowercase{\scshape{t.}} C\lowercase{\scshape{atharines}}, ON L2S3A1, C\lowercase{\scshape{anada}} \\\\
${}^2$D\lowercase{\scshape{epartment}} \lowercase{\scshape{of}} M\lowercase{\scshape{athematics}}\\
U\lowercase{\scshape{niversity of}} C\lowercase{\scshape{\'adiz}}\\
P\lowercase{\scshape{uerto}} R\lowercase{\scshape{eal}}, C\lowercase{\scshape{\'adiz}}, S\lowercase{\scshape{pain}}, 11510\\
}

%\date{}

\begin{abstract}
Dynamical PDEs that have a spatial divergence form possess 
conservation laws that involve an arbitrary function of time.
In one spatial dimension,
such conservation laws are shown to describe
the presence of an $x$-independent source/sink; 
in two and more spatial dimensions,
they are shown to produce a topological charge.
Two applications are demonstrated.
First,
a topological charge gives rise to an associated spatial potential system,
allowing nonlocal conservation laws and symmetries to be found
for a given dynamical PDE.
This type of potential system has a different form and different gauge freedom 
compared to potential systems that arise from ordinary conservation laws. 
Second,
when a topological charge arises from a conservation law whose conserved density is 
non-trivial off of solutions to the dynamical PDE, 
then this relation yields a constraint on initial/boundary data
for which the dynamical PDE will be well-posed.
Several examples of nonlinear PDEs
from applied mathematics and integrable system theory
are used to illustrate these results.
\end{abstract}

\maketitle

\section{Introduction}\label{sec:intro}

Many dynamical PDEs have the form of a spatial divergence. 
In one spatial dimension, the simplest such form consists of 
\begin{equation}\label{1D.eqn}
u_{tx} = D_x F(t,x,u,u_x,u_{xx},\ldots) 
\end{equation}
with $D$ denoting a total derivative.
A prominent physical example is
the Lagrangian form of the Korteweg-de Vries (KdV) equation
for uni-directional shallow water waves,
$u_{tx} +\alpha u_xu_{xx} +\beta u_{xxxx}=0$,
where $u_x=v$ is the wave amplitude. 
(Throughout, $\alpha,\beta$, etc.\ will denote constants.)

In two and three spatial dimensions,
the analogous form (up to a point transformation) is given by 
\begin{equation}\label{multiD.eqn}
\begin{aligned}
u_{tx} = &
D_x F^x(t,x,y,z,u,\nabla u,\nabla^2 u,\ldots) + D_y F^y(t,x,y,z,u,\nabla u,\nabla^2 u,\ldots)
\\&\quad
+ D_z F^z(t,x,y,z,u,\nabla u,\nabla^2 u,\ldots)
\end{aligned}
\end{equation}  
where $\nabla=(\partial_x,\partial_y,\partial_z)$ is the spatial gradient operator.
One important example is the Kadomtsev--Petviashvili (KP) equation \cite{KadPet}
$u_{tx} + (\alpha uu_{x} +\beta u_{xxx})_x \pm u_{yy}=0$,
where $u$ is the wave amplitude. 
In the ``$+$'' case,
it describes shallow water waves with small surface tension, 
and in the ``$-$'' case,
waves in thin films with large surface tension.
Two related examples are 
$(u_{t} + \alpha uu_{x})_x + u_{yy}+u_{zz}=0$
which describes weakly nonlinear, weakly diffracting acoustic waves \cite{ZabKho},  
and $(u_{t} + \beta u^2 u_{x})_x + u_{yy}+u_{zz}=0$
which describes nonlinear, linearly-polarized shear waves \cite{Zab}. 
Another physical example is the Zakharov--Kuznetsov (ZK) equation \cite{ZakKuz}
in Lagrangian form
$u_{tx} +\alpha u_xu_{xx} +\beta u_{xxxx} + \gamma(u_{xxyy} + u_{xxzz})=0$,
which describes ion-acoustic waves in a magnetized plasma,
where $u_x=v$ is the wave amplitude.

In addition,
it is possible to consider PDEs with higher spatial derivatives of $u_t$. 
A notable physical example is the vorticity equation in incompressible fluid flow in two dimensions \cite{MajBer}
$\Delta u_t +u_x \Delta u_y -u_y \Delta u_x =\mu \Delta^2 u$,
where $(-u_y,u_x)$ are the components of the fluid velocity
and $\Delta u$ is the vorticity scalar. 
An example from the theory of integral systems is
the Novikov--Veselov (NV) equation \cite{VesNov},
which has the potential form 
$u_{txy} +\alpha (u_{xy}u_{xx})_x +\beta (u_{xy} u_{yy})_y + u_{xxxxy} + u_{xyyyy} =0$.
This equation arises from isospectral flows
for the two-dimensional Schr\"odinger operator at zero energy \cite{NovVes}. 

All dynamical PDEs \eqref{1D.eqn} and \eqref{multiD.eqn} 
possess at least one family of conservation laws 
involving an arbitrary function $f(t)$ of time $t$:
\begin{equation}\label{multiD.conslaw.f}
D_x( f(t)(u_t - F^x) ) + D_y( {-}f(t)F^y ) + D_z( {-}f(t)F^z ) =0
\end{equation}
holding on solutions $u(t,x,y,z)$. 
This conservation law family arises from $f(t)$ being a multiplier,
whose product with the PDE yields a total divergence. 
A full discussion of multipliers can be found in \Ref{Olv-book,BCA-book,Anc-review}.

When a dynamical PDE can be expressed in a form given by higher spatial derivatives,
then it can admit additional conservation law families involving $f(t)$.
For instance, consider
\begin{equation}\label{multiD.eqn.2ndorder}
\begin{aligned}  
u_{tx} = &
D_x F^x(t,x,y,z,u,\nabla u,\nabla^2 u,\ldots)_x + D_y^2 F^y(t,x,y,z,u,\nabla u,\nabla^2 u,\ldots)
\\&\quad
+ D_z F^z(t,x,y,z,u,\nabla u,\nabla^2 u,\ldots)
\end{aligned}  
\end{equation}
whose form has a second-order derivative with respect to $y$.
Every such dynamical PDE possesses an additional conservation law family 
\begin{equation}\label{multiD.conslaw.yf}
D_x( y f(t)(u_t - F^x) ) + D_y( f(t)F^y -yf(t)(F^y)_y ) +D_z( {-}yf(t)F^z ) =0
\end{equation}
which arises from $yf(t)$ being a multiplier. 
Both the KP equation and the ZK equation are examples.

A converse statement can be made, 
by seeking the most general form for a dynamical PDE that possesses 
a family of conservation laws involving an arbitrary function of $t$. 
As shown recently in \Ref{PopBih}, 
any PDE that admits a multiplier $f(t)$ necessarily has the form of a spatial divergence, 
and any PDE that admits a multiplier $f_0(t)+f_1(t)y$ will be given by 
a spatial divergence form having a second-order derivative with respect to $y$. 
An analogous general form characterizes dynamical PDEs for which 
$f_0(t)+f_1(t)x+f_2(t)y + f_3(t)z$ is a multiplier. 

In general,
conservation laws of dynamical PDEs are central to the analysis of solutions 
by providing physical, conserved quantities as well as conserved norms needed
for studying well-posedness, stability, and global behaviour. 
Conservation laws are also important in checking the accuracy of numerical schemes 
and in devising numerical schemes that have good properties \cite{BihDosPop,WanBihNav}.
For a given dynamical PDE, 
all of its conservation laws (up to any specified differential order)
can be found systematically by the multiplier method \cite{Olv-book,BCA-book,Anc-review}. 
Multipliers that are linear in an arbitrary function $f(t)$ and possibly a finite number of its derivatives
will determine conservation laws in which $f(t)$ and its derivatives appear linearly. 

The wide variety of physically important dynamical PDEs that possess 
conservation laws involving an arbitrary function of $t$ 
motivates asking some fundamental questions about such conservation laws: 
What is their mathematical structure? What physical meaning do they have? 
What information do they provide about solutions of the given dynamical PDE?

The present paper is addressed to these questions and finds an interesting answer:
\begin{itemize}
\item
a non-trivial conservation law involving an arbitrary function of $t$
in one spatial dimension
describes the presence of an $x$-independent \emph{source/sink}; 
\\
\item
in two and more spatial dimensions,
a non-trivial conservation law involving an arbitrary function of $t$
describes a \emph{topological charge}.
\end{itemize}

These main results have some interesting applications for
the study of dynamical PDEs. 

One application is that any non-trivial conservation law involving an arbitrary function $f(t)$
in two or more spatial dimensions 
can be used to introduce a spatial potential system,
allowing nonlocal conservation laws and symmetries to be found
for a given dynamical PDE.
This type of potential system has a different form and different gauge freedom 
compared to potential systems that arise from ordinary conservation laws. 
In general, any potential system can also be useful for doing analysis,
because it provides an equivalent formulation which may have better properties
for the study of solutions of a PDE. 

A more analytical, deeper application emerges when the conserved integral 
on a spatial domain is considered for 
a non-trivial conservation law involving an arbitrary function $f(t)$ 
in two or more spatial dimensions. 
For all solutions of the given dynamical PDE, 
the conserved integral can be shown to be equal to a lower-dimensional integral over the domain boundary,  
which thereby can be evaluated entirely in terms of the boundary values of $u$ and its derivatives. 
This yields an integral constraint relation 
whenever the conserved integral is non-trivial off of the solution space of the dynamical PDE. 
Specifically, if the initial/boundary data is chosen such that the boundary integral vanishes, 
then the conserved integral itself has to vanish. 
This kind of constraint relation has important implications
for the well-posedness of the Cauchy problem, 
especially when the conserved integral is given by a Sobolev norm or an energy.

Both of these applications have not been studied in any generality previously, 
and thus they advance the study of dynamical PDEs as well as the study of conservation laws. 

The rest of the paper is organized as follows.

In section~\ref{sec:results},
we state and prove the main results in a general form for dynamical PDEs in $n\geq 1$ spatial dimensions. 
We also give an explicit explanation of these results
for dynamical PDEs with a spatial-divergence form
in one, two, and three spatial dimensions.

In section~\ref{sec:applications},
we discuss the applications of these results
to constructing spatial potential systems
and to uncovering boundary-value integral relations. 

In section~\ref{sec:examples},
we illustrate all of the preceding developments
by considering some nonlinear PDEs
from applied mathematics and integrable system theory:
the KdV equation in Lagrangian form,
the KP equation,
a universal modified KP equation,
and equations of shear waves.
These examples cover one, two, and three spatial dimensions.
We also consider two examples of nonlinear PDEs with higher spatial derivatives of $u_t$: 
the Novikov--Veselov equation in potential form,
and the vorticity equation. 
For all of these examples, we derive conserved topological charges
from conservation laws that involve an arbitrary function $f(t)$,
and we write down the associated spatial potential systems. 
These conservation laws are found by the standard multiplier method (see e.g.~\Ref{Anc-review}).
In the case of conservation laws whose conserved density is non-trivial off of solutions, 
we derive integral constraint relations on initial data 
for the Cauchy problem.

In section~\ref{sec:conclude},
we make some concluding remarks.

\section{Main results}\label{sec:results}

We begin by considering a general dynamical PDE of order $m\geq 1$
\begin{equation}\label{dyn.pde}
G(t,x,u,\partial u,\ldots,\partial^m u)=0
\end{equation}
for $u(t,x)$ in $n\geq1$ spatial dimensions,
where $t$ is time, and $x=(x^1,\ldots,x^n)$ are spatial coordinates in $\Rnum^n$. 
Here $\partial = (\partial_t,\partial_{x^1},\ldots,\partial_{x^n})$. 
The space of all formal solutions to the PDE will be denoted $\Esp$. 

A \emph{local conservation law} for a PDE \eqref{dyn.pde} 
is a continuity equation
\begin{equation}\label{conslaw}
(D_t \dens + \Div \mathbf\flux)|_\Esp =0
\end{equation}
holding for all solutions $u(t,x)$ of the PDE, 
where $\dens$ is the conserved density,
and $\mathbf\flux=(\flux^1,\ldots,\flux^n)$ is the spatial flux,
which are functions of $t$, $x$, $u$, and derivatives of $u$ up to a finite order. 
The pair $(\dens,\mathbf\flux)$ is called a \emph{conserved current}. 

When solutions $u(t,x)$ are considered 
in a given spatial domain $\Omega\subseteq\Rnum^n$, 
every local conservation law yields a corresponding conserved integral 
\begin{equation}\label{conserved.integral}
\mathcal{C}[u]= \int_{\Omega} \dens|_\Esp d\V
\end{equation}
satisfying the global balance equation
\begin{equation}\label{global.conslaw}
\frac{d}{dt}\mathcal{C}[u]
= -\oint_{\partial\Omega} \mathbf\flux|_\Esp\cdot d\mathbf{\A}
\end{equation}
with $d\mathbf{\A}=\hat{\mathbf n}dA$,
where $\hat{\mathbf n}$ is the unit outward normal vector of the domain boundary $\partial\Omega$, 
and where $dA$ is the boundary volume element. 
This global equation \eqref{global.conslaw} has the physical meaning that
the rate of change of the quantity \eqref{conserved.integral} on the spatial domain 
is balanced by the net outward flux through the boundary of the domain. 

A conservation law is \emph{locally trivial}
\cite{Olv-book,BCA-book,Anc-review} 
when, for all solutions $u(t,x)$ in $\Omega$,
the global balance equation \eqref{global.conslaw} becomes an identity.
This happens iff the conserved density reduces to a spatial divergence,
$\dens|_\Esp =\Div\mathbf\triv|_\Esp$, 
and the spatial flux reduces to a time derivative,
$\mathbf\flux|_\Esp =-D_t\mathbf\triv|_\Esp$ 
modulo a spatial curl, $\Div\underline{\mathbf\trivflux}|_\Esp$
where $\underline{\mathbf\trivflux}$ is a skew-tensor
(namely, $\Div(\Div\underline{\mathbf\trivflux})=0$ holds identically off of $\Esp$).

Likewise, two conservation laws are \emph{locally equivalent}
\cite{Olv-book,BCA-book,Anc-review} 
if they differ by a locally trivial conservation law,
for all solutions $u(t,u)$ in $\Omega$.

A conservation law with $\dens|_\Esp=0$ is called
a \emph{spatial-flux} conservation law,
\begin{equation}\label{fluxconslaw}
\Div\mathbf\flux|_\Esp =0 . 
\end{equation}
Its corresponding conserved integral \eqref{conserved.integral} vanishes.
The meaning of this conserved integral depends on the number of spatial dimensions $n$,
and the topology of the domain $\Omega\subseteq\Rnum^n$. 

In the one-dimensional case ($n=1$), 
when $\Omega$ is a connected interval,
the boundary $\partial\Omega$ consists of two endpoints. 
The global conservation law \eqref{global.conslaw} then reduces to 
boundary source/sink terms 
\begin{equation}\label{1D.sourcesink}
(\flux|_{\partial\Omega})|_\Esp =0 . 
\end{equation}
Its content is that the source/sink flux at one endpoint is balanced by 
the source/sink flux at the other endpoint.
This balance is non-trivial iff $\flux\neq 0$. 

In the multi-dimensional case ($n\geq2$),
there are two different general situations.

If $\Omega$ is a connected volume that is topologically a solid ball, 
then the boundary $\partial\Omega$ consists of
a closed hypersurface $S$ that is topologically a hypersphere.
Hence the conserved integral \eqref{multiD.charge}
shows that the net spatial flux through this hypersurface vanishes, 
\begin{equation}\label{multiD.charge}
\oint_{S} \mathbf\flux|_\Esp\cdot d\mathbf{\A} =0 . 
\end{equation}
This integral describes a vanishing \emph{topological charge},
which is conserved (namely, time-independent).
It is a topological quantity in the sense that it is unchanged 
if the closed hypersurface $S$ is continuously deformed. 

The same interpretation holds if $\Omega$ has a more general topology 
such as a solid torus or its higher genus counterparts. 
This topological charge \eqref{multiD.charge} is analogous to electric/magnetic charge 
for electromagnetic fields in free space, 
where the net electric/magnetic flux through a closed surface vanishes due to the absence of electric/magnetic charges inside the volume bounded by the surface. 

Alternatively, if $\Omega$ is a connected volume that is topologically a shell,
then the boundary $\partial\Omega$ comprises two closed hypersurfaces, $S_1$ and $S_2$, 
each of which is topologically a hypersphere.
The global conservation law \eqref{multiD.charge} thereby
shows that the net spatial flux through two hypersurfaces are equal, 
\begin{equation}\label{multiD.fluxes}
\oint_{S_1} \mathbf\flux|_\Esp\cdot d\mathbf{\A}
=\oint_{S_2} \mathbf\flux|_\Esp\cdot d\mathbf{\A} .
\end{equation}
This equality describes the absence of any source/sink of net flux
inside $\Omega$.
It is analogous to conservation of net electric/magnetic flux 
for electromagnetic fields in free space. 

The same interpretation holds if $\Omega$ is given by
taking some connected volume with a more general topology
and removing from its interior some other connected volume. 

These hypersurface integrals \eqref{multiD.charge} and \eqref{multiD.fluxes}
will be identically zero by Stokes' theorem iff
$\mathbf\flux|_\Esp = \Div\underline{\mathbf\trivflux}|_\Esp$
holds for some skew-tensor $\underline{\mathbf\trivflux}$, 
for all solutions $u(t,x)$ of the PDE. 
Consequently, non-triviality is characterized by the following condition.

\begin{prop}\label{prop:nontriv.flux.multiD}
The topological charge integral \eqref{multiD.charge}
and the source/sink flux integrals \eqref{multiD.fluxes}
are non-trivial iff
$\mathbf\flux|_\Esp \neq \Div\underline{\mathbf\trivflux}|_\Esp$
holds for all skew-tensors $\underline{\mathbf\trivflux}$
which are functions of $t$, $x$, $u$, and derivatives of $u$ up to a finite order. 
\end{prop}

We next state and prove the main results about conservation laws that involve an arbitrary function of time. 

\begin{thm}\label{thm:main}
If a dynamical PDE \eqref{dyn.pde} possesses a local conservation law \eqref{conslaw}
given by a family of conserved currents 
\begin{align}
\dens & =\sum_{i=0}^{N} \dens_i(t,x,u,\partial u,\ldots,\partial^l u) \partial_t^i f(t)
\label{dens.arbfunct}
\\
\mathbf\flux & = \sum_{i=0}^{N+1} {\mathbf\flux}_i(t,x,u,\partial u,\ldots,\partial^l u) \partial_t^i f(t)
\label{flux.arbfunct}
\end{align}
involving an arbitrary function $f(t)$,
where each $\dens_i$ and $\mathbf\flux_i$ does not contain $f(t)$ and its derivatives, 
then the conservation law is locally equivalent to a spatial-flux conservation law 
\begin{equation}\label{chargeflux.conslaw}
f(t)\Div\mathbf\chargeflux|_\Esp =0,
\quad
\mathbf\chargeflux = \sum_{j=0}^{N+1} (-D_t)^{j}\mathbf\flux_{j}(t,x,u,\partial u,\ldots,\partial^l u) .
\end{equation} 
\end{thm}

\begin{proof}
There are three main steps.
We will show first that $T$ is locally trivial,
and next that $\mathbf\flux$ is locally trivial modulo a conserved flux. 
Finally, we will show that the resulting conservation law is locally equivalent to a spatial-flux conservation law. 

The $t$-derivative of the conserved density \eqref{dens.arbfunct} is given by
\begin{equation}\label{Dt.dens}
D_t \dens = \sum_{i=0}^{N} D_t \dens_i \partial_t^i f(t) + \sum_{i=0}^{N} \dens_i \partial_t^{i+1} f(t)
= D_t \dens_0 f(t) + \sum_{i=1}^{N} (D_t\dens_i +\dens_{i-1})\partial_t^i f(t) + \dens_N \partial_t^{N+1} f(t) . 
\end{equation}
Hence, the conservation law \eqref{conslaw} implies
\begin{equation}\label{Div.flux}
(D_t \dens_0+\Div\mathbf\flux_0)|_\Esp f(t) 
+\sum_{i=1}^{N} (D_t\dens_i +\dens_{i-1}+\Div\mathbf\flux_i)|_\Esp \partial_t^i f(t) 
+(\dens_N+\Div\mathbf\flux_{N+1})|_\Esp \partial_t^{N+1} f(t) =0.
\end{equation}
This equation splits with respect to $f(t),\partial_t f(t),\ldots,\partial_t^{N+1} f(t)$, 
because $f(t)$ is an arbitrary function of $t$, 
and so their separate coefficients yield the relations
\begin{align}
\Div\mathbf\flux_{N+1}|_\Esp & = -\dens_{N}|_\Esp , 
\label{dens.N.eqn}\\
\Div\mathbf\flux_i|_\Esp & = -(D_t \dens_i + \dens_{i-1})|_\Esp, 
\quad
i=1,\ldots,N-1, 
\label{dens.i.eqn}\\
\Div\mathbf\flux_0|_\Esp & = -D_t \dens_0|_\Esp . 
\label{flux.0.eqn}
\end{align}
Equations \eqref{dens.N.eqn} and \eqref{dens.i.eqn}
can be solved recursively for $\dens_N,\ldots,\dens_0$:
\begin{equation}\label{dens.i}
\dens_i|_\Esp = -\sum_{j=0}^{N-i} (-D_t)^j\Div\mathbf\flux_{i+j+1}|_\Esp, 
\quad
i=0,\ldots,N . 
\end{equation}
This shows that each $\dens_i$ is locally trivial,
namely
\begin{equation}\label{dens.triv.i}
\dens_i|_\Esp = \Div\mathbf\triv_i|_\Esp,
\end{equation}  
with
\begin{equation}\label{triv.i}
\mathbf\triv_i= -\sum_{j=0}^{N-i} (-D_t)^j\mathbf\flux_{i+j+1},
\quad
i=0,\ldots,N . 
\end{equation}
Therefore,
\begin{equation}\label{dens.triv}
\dens|_\Esp = \Div\mathbf\triv|_\Esp,
\quad
\mathbf\triv = -\sum_{i=0}^{N} \partial_t^i f(t) \sum_{j=0}^{N-i} (-D_t)^j \mathbf\flux_{i+j+1}
\end{equation}  
is locally trivial.

Next, the $t$-derivative of $\mathbf\triv$ is given by 
\begin{equation}\label{Dt.triv}
\begin{aligned}
D_t\mathbf\triv 
& = \sum_{i=0}^{N} \big(
\partial_t^{i+1} f(t)\, \mathbf\triv_{i} + \partial_t^{i} f(t) D_t\mathbf\triv_{i}
\big)
\\
&
= f(t) D_t\mathbf\triv_{0}
+ \sum_{i=1}^{N} \partial_t^{i} f(t)(D_t\mathbf\triv_{i} + \mathbf\triv_{i-1})
+\partial_t^{N+1}f(t)\, \mathbf\triv_{N} , 
\end{aligned}
\end{equation}
which can be written in terms of the expressions $\mathbf\flux_i$.
In the last term in \eqref{Dt.triv}, 
\begin{equation}\label{last.term}
\mathbf\triv_{N}
= -\mathbf\flux_{N+1}
\end{equation}
holds from expression \eqref{triv.i} for $i=N$. 
Likewise the first term in \eqref{Dt.triv} can be simplified by observing 
that 
$D_t\mathbf\triv_{0}
= \sum_{j=0}^{N} (-D_t)^{j+1}\mathbf\flux_{j+1}$
holds from expression \eqref{triv.i} for $i=0$,
and that 
$\sum_{j=0}^{N} (-D_t)^{j+1}\Div\mathbf\flux_{j+1}|_\Esp = -\Div \mathbf\flux_0|_\Esp$
follows from equation \eqref{flux.0.eqn} combined with equation \eqref{dens.i} for $i=0$. 
This implies 
\begin{equation}\label{first.term}
D_t\mathbf\triv_{0}
= \sum_{j=0}^{N} (-D_t)^{j+1}\mathbf\flux_{j+1}
= (\mathbf\chargeflux -\mathbf\flux_0), 
\quad
\Div\mathbf\chargeflux|_\Esp=0 ,
\end{equation}
holds for some vector function $\mathbf\chargeflux(t,x,u,\partial u,\ldots,\partial^{l'} u)$.
For the remaining terms in \eqref{Dt.triv}, 
\begin{equation}\label{middle.terms}
(D_t\mathbf\triv_{i} + \mathbf\triv_{i-1})
= \sum_{j=0}^{N-i} (-D_t)^{j+1}\mathbf\flux_{i+j+1}
-\sum_{j=0}^{N-i+1} (-D_t)^{j}\mathbf\flux_{i+j}
= -\mathbf\flux_i ,
\quad
i=1,\ldots,N
\end{equation}
reduces to a telescoping sum.
Substituting equations \eqref{last.term}--\eqref{middle.terms}
into equation \eqref{Dt.triv}
yields 
\begin{equation}
D_t\mathbf\triv = 
f(t)\mathbf\chargeflux -\sum_{i=0}^{N+1} \partial_t^i f(t) \mathbf\flux_i . 
\end{equation}
This equation combined with the flux expression \eqref{flux.arbfunct}
then gives 
\begin{equation}\label{flux.triv}
\mathbf\flux|_\Esp = -D_t\mathbf\triv|_\Esp  + f(t)\mathbf\chargeflux|_\Esp , 
\end{equation}
which consists of a locally trivial term plus a divergence-free term. 

Finally,
consider the locally trivial conserved current
$(\Div\mathbf\triv,-D_t\mathbf\triv)|_\Esp$. 
When this conserved current is added to the conserved current
given by the density \eqref{dens.triv} and the flux \eqref{flux.triv},
the resulting locally equivalent conservation law is given by the conserved current
\begin{equation}
(0,f(t)\mathbf\chargeflux)|_\Esp 
\end{equation}  
as shown by equations \eqref{dens.triv} and \eqref{flux.triv}. 
To conclude the proof, 
note that equation \eqref{first.term} yields the expression for $\mathbf\chargeflux$ 
given in equation \eqref{chargeflux.conslaw}. 
\end{proof}

Existence of a spatial-flux conserved current \eqref{chargeflux.conslaw} will require 
that a dynamical PDE \eqref{dyn.pde} have a certain form. 
The problem of exactly determining this form can be addressed by 
first working out the general form for the multiplier, 
which must look like 
$\sum_{i=0}^{N'} \partial_t ^i f(t) Q_i(t,x,u,\partial u,\ldots,\partial^{l'_i} u)$, 
and then generalizing the methods in \Ref{PopBih} using Euler operators 
for determining the form of PDEs that possess multipliers $f(t)$. 
This is a non-trivial problem and will be left for elsewhere. 

It will be useful for the sequel to remark that
the underlying mathematical setting for all the developments here
is calculus on jet space \cite{Olv-book}.
For a given PDE \eqref{dyn.pde}, the associated jet space is simply
the coordinate space $J= (t,x,u,\partial u,\partial^2 u,\ldots)$.
The solution space $\Esp$ of the PDE is represented by the surface
defined by equation \eqref{dyn.pde} in $J$
along with the corresponding surfaces given by all derivatives of equation \eqref{dyn.pde}.
Each solution $u(t,x)$ of the PDE is represented by a point on these surfaces. 
When an expression such as $\dens$ or $\mathbf\flux$ is evaluated on the solution space $\Esp$,
this means that the expression $\dens|_\Esp$ or $\mathbf\flux|_\Esp$ is evaluated on the solution surfaces in $J$.
For purposes of computation,
the evaluation of any expression on $\Esp$ can be carried out by
first writing the PDE \eqref{dyn.pde} in a solved form
with respect to some leading derivative, 
and then substituting the leading derivative and all of its derivatives into the expression.
See \Ref{Anc-review} for more details and examples.

\subsection{Topological charges for dynamical PDEs with a spatial-divergence form}\label{sec:topologicalcharge}

As as a consequence of Theorem~\ref{thm:main}, 
the following result is obtained. 

\begin{cor}\label{cor:topological.charge}
For a dynamical PDE \eqref{dyn.pde}  in $n>1$ spatial dimensions, 
any conservation law that involves an arbitrary function of $t$
will yield a corresponding conserved topological charge
\begin{equation}
\oint_{\partial\Omega} \mathbf\chargeflux|_\Esp\cdot d\mathbf{\A} =0
\end{equation}
associated to any spatial domain $\Omega\subseteq\Rnum^n$
whose boundary is a closed hypersurface $\partial\Omega$. 
The topological nature of the charge is that it is unchanged 
under continuous deformations of the hypersurface. 
\end{cor}

This result will now be specialized to dynamical PDEs with a spatial-divergence form
\begin{equation}\label{pde.divform}
\k\cdot\nabla u_t = \nabla\cdot {\mathbf F}(t,x,u,\nabla u,\ldots,\nabla^m u)
\end{equation}
where 
$\mathbf{F}= (F^1,\ldots,F^n)$ is a vector function,
and $\k=(k^1,\ldots,k^n)$ is a constant unit vector. 
Any such PDE can be directly expressed as a spatial-flux conservation law
\begin{equation}\label{nD.flux.conslaw}
\nabla\cdot\mathbf\chargeflux|_\Esp =0,
\quad
\mathbf\chargeflux= u_t \k - \mathbf{F} . 
%(k^1 u_t -F^1,\ldots,k^n u_t -F^n) 
\end{equation}  
This conservation law is non-trivial if and only if 
the flux does not have the form of a spatial curl, 
$\mathbf\chargeflux|_\Esp \neq \Div\underline{\mathbf\trivflux}|_\Esp$, 
for all skew-tensor functions
$\underline{\mathbf\trivflux}(t,x,u,\partial u,\ldots,\partial^l u)$.
But the relation $u_t\k -\mathbf{F}=\Div\underline{\mathbf\trivflux}|_\Esp$ 
would clearly be inconsistent in any finite jet space
because the left side contains $u_t$ but no derivatives of $u_t$,
while if $\underline{\mathbf\trivflux}|_\Esp$ contains $u_t$ (or its derivatives) 
then the right side will always contain at least $\nabla u_{t}$ (or its derivatives). 
This argument can be made rigorous by employing 
the spatial Euler operator \cite{BCA-book,Anc-review}
similarly to the methods in \Ref{PopBih}.

We will now look at the content of the resulting global conservation law
first for PDEs in one spatial dimension
and then for PDEs in two and more spatial dimensions. 

Consider a dynamical PDE \eqref{pde.divform} in one spatial dimension,
$u_{tx} = D_x F(t,x,u,\partial_x u,\ldots,\partial_x^m u)$. 
Every such PDE has the form of a spatial-flux conservation law $D_x\flux|_\Esp=0$
in which the flux is given by $\flux = u_t - F$.
For solutions $u(t,x)$ on a connected interval $\Omega\subseteq\Rnum$,
the corresponding global form of the conservation law is obtained by
integration with respect to $x$,
yielding
\begin{equation}\label{1D.sourcesink.conslaw}
\begin{aligned}
& \int_{\Omega} D_x(u_t-F)|_\Esp\,dx = 0 \\&\quad
=((u_t-F)|_{\partial\Omega})|_\Esp
\end{aligned}
\end{equation}
which is a boundary term at the two endpoints $\partial\Omega$ of the interval. 
This boundary-type conservation law corresponds to the integrated form of the PDE
\begin{equation}\label{1D.ut}
u_{t} = F(t,x,u,\partial_x u,\ldots,\partial_x^m u) + w(t)
\end{equation}
whereby $u_t-F = w$ is independent of $x$.
Here $w$ can be viewed as a source/sink
in the resulting evolution equation \eqref{1D.ut}.
Note that, through this equation,
each solution $u(t,x)$ of the PDE
$u_{tx} = D_x F(t,x,u,\partial_x u,\ldots,\partial_x^m u)$
will give rise to a corresponding function $w(t)$
which, in general, will be non-zero. 
As a consequence,
the source/sink conservation law \eqref{1D.sourcesink.conslaw}
is non-trivial. 

The situation differs for a dynamical PDE \eqref{pde.divform} 
in two spatial dimensions.
First, note that we can put $\k=(1,0)$ by a point transformation,
so that the PDE takes the specific form 
\begin{equation}\label{2D.eqn}
u_{tx} = \nabla_x F^x(t,x,y,u,\nabla u,\ldots,\nabla^m u) + \nabla_y F^y(t,x,y,u,\nabla u,\ldots,\nabla^m u) . 
\end{equation}
This is a spatial-flux conservation law
$(D_x\flux^x +D_y\flux^y)|_\Esp=0$ 
in which the flux is given by the vector function
\begin{equation}\label{2D.flux}
(\flux^x,\flux^y) = (u_t - F^x,-F^y) . 
\end{equation}
The corresponding global form of the conservation law is obtained by
integration with respect to $x,y$ over a given connected domain $\Omega\subseteq\Rnum^2$,
yielding
\begin{equation}\label{2D.charge.conslaw}
\begin{aligned}
& \int_{\Omega} (D_x(u_t - F^x) + D_y(-F^y))|_\Esp \,dxdy =0 \\&\quad
= \oint_{\partial\Omega} F^y\,dx + (u_t - F^x)\,dy|_\Esp 
\end{aligned}
\end{equation}
for all solutions $u(t,x)$. 
This line integral is the two-dimensional form of the topological charge \eqref{multiD.charge}.
It is non-trivial because the vector function \eqref{2D.flux}
cannot be expressed in a curl form $(D_y\trivflux,-D_x\trivflux)$
for any scalar function $\trivflux$ of $t$, $x$, $y$, $u$, and derivatives of $u$ up to a finite order,
for all solutions of the PDE \eqref{2D.eqn}. 
(This follows from the argument used to show non-triviality of the conservation law \eqref{nD.flux.conslaw}.)

The topological charge conservation law \eqref{2D.charge.conslaw}
can also be viewed as corresponding to the integrated form of the PDE \eqref{2D.eqn}
as given by
\begin{equation}\label{2D.integrated.eqn}
u_t-F^x(t,x,y,u,\nabla u,\ldots,\nabla^m u) = w_y,
\quad
F^y(t,x,y,u,\nabla u,\ldots,\nabla^m u) =w_x . 
\end{equation}
Note that each solution $u(t,x,y)$ of the PDE \eqref{2D.eqn}
will give rise to a corresponding function $w(t,x,y)$. 
In particular, $w$ depends nonlocally on $u$,
and so the equations \eqref{2D.integrated.eqn}
cannot be viewed as surfaces in the jet space $J= (t,x,u,\partial u,\partial^2 u,\ldots)$.
In terms of $w$, 
the topological charge conservation law \eqref{2D.charge.conslaw}
takes the form $\oint_{\partial\Omega} w_x\,dx + w_y\,dy=0$
which holds for any function $w(t,x,y)$
by the gradient line integral theorem. 

There is also a dynamical interpretation of the topological charge conservation law \eqref{2D.charge.conslaw},
which comes from expressing it in the form 
\begin{equation}\label{2D.dyn.flux.conslaw}
\frac{d}{dt}\oint_{C} u\,dy|_\Esp 
= \oint_{C} (F^x\,dy - F^y\,dx)|_\Esp 
\end{equation}
where $C$ is any closed curve in the $(x,y)$-plane. 
The line integral on the left side is the net circulation $\oint_{C} u\,dy|_\Esp$ of 
the transverse transport vector $(0,u)= u\k_\perp$ around the closed curve, 
while the line integral on the right side is the net circulation of 
the vector field $(-F^y,F^x)|_\Esp$. 
Thus, this vector field drives the rate of change of the net circulation of 
the transport vector. 

In three or more spatial dimensions, 
the global form of the spatial-flux conservation law \eqref{nD.flux.conslaw}
is given by the vanishing flux integral 
\begin{equation}\label{fluxcharge.conslaw}
\oint_{S} (\k u_t -\mathbf{F})|_\Esp\cdot d\mathbf{\A} =0
\end{equation}
on any closed hypersurface $S$ in $\Rnum^n$.
This integral describes a conserved topological charge
for all solutions $u(t,x)$ of the PDE \eqref{pde.divform}. 
It can be interpreted dynamically as stating that 
the rate of change of the net flux of the transport vector $u \k$ 
through a closed hypersurface 
is balanced by the net flux of the vector field $\mathbf{F}$:
\begin{equation}\label{dyn.flux.conslaw}
\frac{d}{dt} \oint_{S} u \k|_\Esp \cdot d\mathbf{\A} 
= \oint_{S} \mathbf{F}|_\Esp \cdot d\mathbf{\A} . 
\end{equation}

Further understanding will be provided by the examples in Section~\ref{sec:examples}.
A general discussion of topological conservation laws 
in three dimensions can be found in \Ref{AncChe2018}.

\section{Applications}\label{sec:applications}

We will present two innovative applications of the results in Theorem~\ref{thm:main} and Corollary~\ref{cor:topological.charge}:\\
\indent$\bullet$
constructing spatial potential systems; \\
\indent$\bullet$
uncovering integral constraint relations on solutions. \\
For each application, we first consider a general dynamical PDE \eqref{dyn.pde},
and then we will specialize the presentation to the situation
when the PDE has a spatial-divergence form \eqref{pde.divform}. 
Examples will be given in section~\ref{sec:examples}.

\subsection{Potential systems}

Any spatial-flux conservation law \eqref{chargeflux.conslaw}
for a given dynamical PDE \eqref{dyn.pde} in $n>1$ spatial dimensions
can be converted into a spatial potential system
via the introduction of a potential $\underline{\mathbf w}(t,x)$
in the form of a skew-tensor. 
This is achieved by expressing the spatial-flux vector in a curl form 
\begin{equation}\label{div.curl.flux}
\mathbf\chargeflux(t,x,u,\partial u,\ldots) =\Div\underline{\mathbf w}
\end{equation}
in terms of the potential.
Then $\Div\mathbf\chargeflux(t,x,u,\partial u,\ldots) =0$
holds identically (off of the solution space $\Esp$ of the PDE),
as $\Div(\Div\underline{\mathbf w})=\Div^2\underline{\mathbf w}$
is identically zero
due to the Hessian matrix $\Div^2$ being symmetric
while the potential $\underline{\mathbf w}$ is skew. 
Note this is the $n$-dimensional version of the identity that 
the divergence of a curl in three dimensions is identically zero. 

We view the vector-flux equation \eqref{div.curl.flux}
as defining a potential system with variables $(u,\underline{\mathbf w})$ 
associated to the given dynamical PDE \eqref{dyn.pde}.
Such a potential system is purely spatial, 
since it does not contain a time derivative of $\underline{\mathbf w}$, 
in contrast to the more familiar potential systems that arise from 
ordinary conservation laws in which the conserved density is non-trivial
(see e.g. \Ref{BCA-book}). 
Moreover, it exists only for $n>1$, 
since curls do not exist in dimension $n=1$. 

There is a natural correspondence between
the solution space of the dynamical PDE
and the solution space of the spatial potential system 
modulo some gauge freedom which we will describe next. 
This gauge freedom is different in form than the gauge freedom inherent in non-spatial potential systems. 

Consider, firstly, the situation in two spatial dimensions ($n=2)$.
The spatial-flux vector $\mathbf\chargeflux=(\chargeflux^x,\chargeflux^y)$
has two components,
while the potential can be expressed explicitly as a skew matrix
\begin{equation}
  \underline{\mathbf w}=\begin{pmatrix}0&w\\-w&0\end{pmatrix}
\end{equation}  
involving a single scalar variable $w$. 
The spatial potential system \eqref{div.curl.flux} is given by 
\begin{equation}\label{2D.div.curl.flux}
\chargeflux^x = w_y,
\quad
\chargeflux^y = -w_x . 
\end{equation}
It is unchanged under the gauge freedom
\begin{equation}\label{2D.gaugefreedom}
w\to w+ \chi(t)
\end{equation}
where $\chi(t)$ is an arbitrary differentiable function of $t$.
Every solution $(u(t,x),w(t,x))$ of this potential system \eqref{2D.div.curl.flux}
yields a solution $u(t,x)$ of the given dynamical PDE.
Conversely, every solution $u(t,x)$ of the given dynamical PDE
determines a set of solutions $(u(t,x),w(t,x)+\chi(t))$ of the potential system.
Modulo the gauge freedom, this correspondence is one-to-one.

For a dynamical PDE that has a spatial-divergence form \eqref{2D.eqn}
in two spatial dimensions, 
the spatial potential system \eqref{2D.div.curl.flux} is given by
the pair of equations \eqref{2D.integrated.eqn} in terms of $(u,w)$. 

By comparison, non-spatial potential systems in two spatial dimensions
involve three potential variables 
and have a gauge freedom which involves 
arbitrary functions of all three independent variables $t,x,y$. 
Such gauge freedom requires that a gauge condition be appended to the potential system 
so that it can yield potential symmetries and potential conservation laws \cite{AncBlu1997b,BCA-book}. 
In contrast, 
a spatial potential system \eqref{2D.div.curl.flux} has less gauge freedom, 
and therefore gauge conditions are not necessarily required. 

Secondly, consider the situation in three spatial dimensions ($n=3$).
Now the spatial-flux vector has three components
$\mathbf\chargeflux=(\chargeflux^x,\chargeflux^y,\chargeflux^z)$,
and the potential can be expressed explicitly as a $3\times 3$ skew matrix
\begin{equation}
\underline{\mathbf w}=\begin{pmatrix}0&w^z&-w^y\\-w^z&0&w^x\\w^y&-w^x&0\end{pmatrix}
\end{equation}
which involves three scalar variables $(w^x,w^y,w^z)$. 
The resulting spatial potential system \eqref{div.curl.flux} has the form 
\begin{equation}\label{3D.div.curl.flux}
\chargeflux^x=(w^z)_y-(w^y)_z,
\quad
\chargeflux^y=(w^x)_z-(w^z)_x,
\quad
\chargeflux^z=(w^y)_x-(w^x)_y . 
\end{equation}
It is unchanged under the gauge freedom
\begin{equation}\label{3D.gaugefreedom}
(w^x,w^y,w^z)\to(w^x+\partial_x\chi(t,x,y,z),w^y+\partial_y\chi(t,x,y,z),w^z+\partial_z\chi(t,x,y,z))
\end{equation}
where $\chi(t,x,y,z)$ is an arbitrary differentiable function of all independent variables.
Modulo this gauge freedom, there is a one-to-one correspondence between 
the solution space of the potential system and the solution space of the given dynamical PDE.

When a dynamical PDE in three dimensions 
has a spatial-divergence form
\begin{equation}\label{3D.eqn}
\begin{aligned}
u_{tx} &
= \nabla_x F^x(t,x,y,z,u,\nabla u,\ldots,\nabla^m u) + \nabla_y F^y(t,x,y,z,u,\nabla u,\ldots,\nabla^m u)
\\&\qquad 
+ \nabla_z F^z(t,x,y,z,u,\nabla u,\ldots,\nabla^m u) , 
\end{aligned}
\end{equation}
its spatial potential system \eqref{3D.div.curl.flux} is given by
\begin{subequations}\label{3D.integrated.eqn}
\begin{align}
u_t -F^x(t,x,y,z,u,\nabla u,\ldots,\nabla^m u) & = (w^z)_y-(w^y)_z,
\\
F^y(t,x,y,z,u,\nabla u,\ldots,\nabla^m u) & = (w^x)_z-(w^z)_x,
\\
F^z(t,x,y,z,u,\nabla u,\ldots,\nabla^m u) & = (w^y)_x-(w^x)_y, 
\end{align}
\end{subequations}
in terms of $(u,w^x,w^y,w^z)$. 

Non-spatial potential systems in three spatial dimensions involve six potentials and four gauge functions. 
Thus, spatial potential systems \eqref{3D.div.curl.flux} are considerably easier to use
(in particular, they require only a single gauge condition). 

A similar correspondence extends to $n>3$ spatial dimensions,
where the spatial potential system \eqref{div.curl.flux} for a dynamical PDE
with a spatial-divergence form \eqref{pde.divform}
is given by 
\begin{equation}
\k u_t -{\mathbf F}(t,x,u,\nabla u,\ldots,\nabla^m u)
=\Div\underline{\mathbf w}
\end{equation}
for $(u,\underline{\mathbf w})$. 
The gauge freedom in the skew-tensor variable consists of
\begin{equation}
\underline{\mathbf w}(t,x) \to\underline{\mathbf w}(t,x) + \Div\gaugevar(t,x)
\end{equation}
with $\gaugevar(t,x)$ being an antisymmetric rank-3 tensor
whose components are arbitrary differentiable functions of all independent variables $(t,x)$.
In particular,
note that $\Div\Div\gaugevar(t,x)=\Div^2\gaugevar(t,x)$ is identically zero
since $\Div^2$ is symmetric while $\gaugevar$ is totally antisymmetric.

\subsection{Integral constraint relations}

For a given dynamical PDE \eqref{dyn.pde} in $n>1$ spatial dimensions, 
any conservation law \eqref{dens.arbfunct}--\eqref{flux.arbfunct}
involving an arbitrary function $f(t)$ and its derivatives $\partial_t^i f(t)$,
$i=0,1,\ldots,N$, 
gives rise to the set of $N+1$ divergence relations \eqref{dens.triv.i}--\eqref{triv.i}
which hold for all solutions $u(t,x)$ of the PDE.
Each of these relations can be integrated
over any given spatial domain $\Omega\subseteq\Rnum^n$ to obtain
an associated integral relation 
\begin{equation}\label{integral.relation}
\int_{\Omega} \dens_i|_\Esp\,d\V
= \oint_{\partial\Omega} \mathbf\triv_i|_\Esp\cdot d\mathbf\A, 
\quad
i=0,\ldots,N
\end{equation}  
where 
$\dens_i$ is a scalar function of $t$, $x$, $u$, and derivatives of $u$,
and where 
$\mathbf\triv_i$ is a vector function of $t$, $x$, $u$, and derivatives of $u$. 
This relation \eqref{integral.relation}
holds for all solutions $u(t,x)$ of the PDE.
It shows that the volume integral on the left-hand side
can be evaluated entirely in terms of the values of $u$ and its derivatives
on the domain boundary
appearing in the hypersurface flux integral on the righthand side.

We remark that, off of the solution space of a given PDE,
an integral relation \eqref{integral.relation}
is equivalent to a divergence-type identity
\begin{equation}\label{dens.div.id}
\dens = \Div \mathbf\triv + R(G)
\end{equation}
where $R$ is a linear differential operator in total derivatives 
whose coefficients are non-singular on the solution space $\Esp$ of the PDE.

An integral relation \eqref{integral.relation} is significant
when the Cauchy problem for a given dynamical PDE is considered
on a spatial domain $\Omega\subseteq\Rnum^n$.
For simplicity, suppose that the PDE has a spatial divergence form \eqref{pde.divform},
whereby it can be expressed as a nonlocal evolution equation 
$u_t = \k\cdot{\mathbf F} + (\k\cdot\nabla)^{-1}{\mathbf F}_\perp$
where $\k\cdot{\mathbf F}_\perp=0$. 
The Cauchy problem then consists of specifying
initial data $u_0(x)$ for $u$ at $t=0$ on $\Omega$.
There are two basic situations:
$\Omega\subset\Rnum^n$ is finite;
or $\Omega=\Rnum^n$ is infinite. 

When the domain is finite, 
if boundary conditions for $u$ are posed on $\partial\Omega$ for $t\geq0$
such that $\oint_{\partial\Omega} \mathbf\triv\cdot d\mathbf\A=0$,
then the initial data must satisfy the restriction
\begin{equation}\label{u0.finitedomain}
\int_{\Omega} \dens|_{u=u_0}\,d\V =0 . 
\end{equation}
Likewise, when domain is infinite,
if $u$ obeys asymptotic conditions posed for $t\geq0$
such that $\lim_{r\to\infty}\oint_{S^{n-1}(r)} \mathbf\triv\cdot d\mathbf\A=0$
where $S^{n-1}(r)$ is a sphere of radius $r>0$ in $\Rnum^n$, 
then the initial data must satisfy the restriction
\begin{equation}\label{u0.infinitedomain}
\int_{\Rnum^n} \dens|_{u=u_0}\,d\V =0 . 
\end{equation}

Such a restriction \eqref{u0.finitedomain} or \eqref{u0.infinitedomain}
has significant implications if the integral is
a non-negative energy expression
or a norm in a function space such as $L^p$ or $H^s$. 
In this situation, 
the integral would be a priori non-negative for all non-trivial solutions
$u(t,x)$ of the given dynamical PDE on the spatial domain $\Omega$,
which seems to suggest a contradiction. 
But in fact it implies that there must exist some a priori restrictions
on the possible kinds of initial conditions and boundary conditions 
that can be posed to have the resulting initial-boundary value problem for the PDE
be well-posed in the function space determined by the norm
$\int_\Omega T|_\Esp\,dV$.

To illustrate in general how such a non-trivial integral relation could exist,
consider a dynamical PDE with a spatial-divergence form \eqref{2D.eqn}
in two dimensions:
$G = u_{tx} -D_x F^x -D_y F^y=0$. 
For simplicity, suppose that the PDE possesses a scaling symmetry
$t\to\lambda t$, $x\to \lambda^a x$, $y\to \lambda^b y$, $u\to \lambda^c u$,
with scaling weights $a,b,c$. 
We can derive an integral relation for $T=u_x^{p+1}$, 
with $p$ assumed to be a positive integer, 
by the following steps. 

First, we start from a scaling-homogeneous multiplier of the form
\begin{equation}
  Q=f(t) u_x^p + f'(t) y
\end{equation}  
where the constants $p,a,b,c$ are related by $(c-a)p=b-1$ 
due to scaling homogeneity.
Next, we observe 
\begin{equation}
\begin{aligned}
QG & =(f(t) u_x^p + f'(t) y)(u_{tx} -D_xF^x -D_yF^y)
\\
& = D_t( \tfrac{1}{p+1} u_x^{p+1}f(t) ) +D_x( y (u_t-F^x)f'(t) ) -D_y( \tfrac{1}{2p+1} u_x^{2p+1} f(t) +y F^y f'(t) )
\\&\qquad
- u_x^p(D_xF^x +D_y\tilde F^y)f(t) +\tilde F^y f'(t) ,
\quad
\tilde F^y = F^y -\tfrac{1}{p+1} u_x^{p+1} .
\end{aligned}  
\end{equation}
This identity shows that $Q$ will be multiplier iff the last two terms are separately a total spatial derivative
\begin{equation}
u_x^p(D_xF^x +D_y\tilde F^y) = D_x X_0 + D_y Y_0,
\quad
\tilde F^y = D_x X_1 + D_y Y_1 , 
\end{equation}
for some functions $X_0$, $X_1$, $Y_0$, $Y_1$
of $u$ and its spatial derivatives as well as possibly $t,x,y$.
The necessary and sufficient conditions are given by
\begin{equation}
E_u(u_x^{p-1}(u_{xx}F^x +u_{xy}\tilde F^y)) =0, 
\quad
E_u(\tilde F^y)=0 ,%  = -pu_x^{p-1} (D_x^2 F^x +D_xD_y F^y) - p(F^x)'{}^*(u_x^{p-1}u_{xx}) - p(F^y)'{}^*(u_x^{p-1}u_{xy})
\end{equation}  
where $E_u$ is the Euler operator (variational derivative) with respect to $u$.
If we specify a maximum differential order for $F^x$ and $\tilde F^y$
as functions in jet space,
then this pair of equations splits with respect to all higher-order jet variables,
yielding a PDE system on $F^x$ and $\tilde F^y$.
The system turns out to have non-trivial solutions
as shown by the example of the mKP equation
in section~\ref{sec:examples}.

Hence, we obtain a conservation law of the form \eqref{dens.arbfunct} and \eqref{flux.arbfunct}
with
\begin{subequations}
\begin{gather}
T_0  =\tfrac{1}{p+1} u_x^{p+1},
\\
\flux^x_0 = -X_0, 
\quad
\flux^x_1 = y(u_t-F^x) +X_1 , 
\\
\flux^y_0 = -\tfrac{1}{2p+1} u_x^{2p+1} -Y_0,
\quad
\flux^y_1 = -y F^y +Y_1 , 
\end{gather}
\end{subequations}
where $N=0$. 
Finally, from Theorem~\ref{thm:main},
we obtain a single spatial divergence relation \eqref{dens.triv.i}--\eqref{triv.i}, 
which is explicitly given by 
\begin{equation}
\begin{aligned}
T_0|_\Esp & = \tfrac{1}{p+1} u_x^{p+1} \\
& = -(D_x\flux^x_1 + D_y\flux^y_1)|_\Esp = D_x( y(F^x-u_t)-X_1 )|_\Esp +D_y( y F^y -Y_1 )|_\Esp  . 
\end{aligned}
\end{equation}  
Such an identity is highly non-trivial because it expresses a power of $u_x$ as a local divergence. 
This yields a corresponding non-trivial integral relation
\begin{equation}
\int_{\Omega} \tfrac{1}{p+1} u_x^{p+1}\,d\V
= \oint_{\partial\Omega} ( y(F^x-u_t)-X_1 )|_\Esp dy +( Y_1 - y F^y)|_\Esp dx
\end{equation}    
where $\Omega\subseteq \Rnum^2$ is a domain in the $(x,y)$-plane,
and $\partial\Omega$ is its boundary curve. 
If $p$ is an odd integer,
then we conclude that the $L^{p+1}(\Omega)$ norm of $u_x$
is determined by the values of $u$ and derivatives of $u$ on the boundary. 
In particular,
this holds for any initial data posed for the given PDE \eqref{2D.eqn}. 
As a consequence,
for initial data on $\Omega=\Rnum^2$
with asymptotic decay conditions at $\partial\Omega=\lim_{r\to\infty} S^1(r)$,
the Cauchy problem for $u_t = F^x + \partial_x^{-1}(D_yF^y)$
cannot be well-posed when the initial data have
sufficiently rapid radial decrease such that
%$\lim_{r\to\infty}\oint_{S^1(r)} (x(y(F^x-u_t)-X_1) -y(Y_1 - y F^y))|_\Esp r\,d\theta$
\begin{equation}
\| u_x \|_{L^{p+1}}<\infty
\text{ and } 
\lim_{r\to\infty}\oint_{S^1(r)} \big( (X_1+r(\partial_x^{-1}(D_yF^y)\sin\theta)\cos\theta +(Y_1-rF^y\sin\theta)\sin\theta \big)|_\Esp r\,d\theta =0
.
\end{equation}

In general, 
the connection established by equation \eqref{integral.relation} relating 
integral constraints \eqref{u0.finitedomain}--\eqref{u0.infinitedomain}
and conservation laws \eqref{dens.arbfunct}--\eqref{flux.arbfunct} involving an arbitrary function of $t$
provides a systematic way to detect and to find 
integral constraints and corresponding local differential identities.

\section{Examples}\label{sec:examples}

We will consider several examples of PDEs with a spatial divergence form
in one, two, and three spatial dimensions.
The examples come from applied mathematics and integrable system theory. 
For each PDE,
we first present conservation laws that involve an arbitrary function $f(t)$.
These conservation laws are obtained by the standard multiplier method, 
explained in \Ref{Olv-book,BCA-book,Anc-review}. 
We then write down the topological charge(s) and discuss the applications 
explained in section~\ref{sec:applications}.

\subsection{KdV equation}

The KdV equation $v_t + vv_x + v_{xxx}=0$, 
where the coefficients have been scaled to $1$,
is a well-known integrable system that describes uni-directional shallow water waves.
It has a Lagrangian formulation in terms of a potential $u$
given by $v=u_x$:
\begin{equation}\label{pKdV}
u_{tx} + u_xu_{xx} + u_{xxxx}=0 . 
\end{equation}
The potential has gauge freedom $u\to u+f(t)$ where $f(t)$ is arbitrary.

Equation \eqref{pKdV} matches the form of the general spatial divergence PDE \eqref{pde.divform}
for $n=1$ with 
\begin{equation}
F=-(\tfrac{1}{2}u_x^2 +u_{xxx}) . 
\end{equation}

Through Noether's theorem, 
all local conservation laws of equation \eqref{pKdV} arise from multipliers $Q$,
which coincide with the characteristics of variational symmetries.
The gauge freedom in $u$ corresponds to a variational point symmetry
and hence a multiplier 
\begin{equation}
Q=f(t) . 
\end{equation}
The corresponding conservation law is given by the conserved current
\begin{equation}
(\dens,\flux) = (0, (u_t + \tfrac{1}{2}u_x^2 + u_{xxx})f(t) ) . 
\end{equation}
This represents a spatial-flux conservation law.
Its global form on any connected interval $[a,b]$ consists of 
$f(t)(u_t + \tfrac{1}{2}u_x^2 + u_{xxx})|^{b}_{a} =0$
for all KdV solutions $u(t,x)$. 
The overall factor $f(t)$ can be dropped without loss of generality,
yielding the global conservation law
\begin{equation}
(u_t + \tfrac{1}{2}u_x^2 + u_{xxx})|^{b}_{a} =0
\end{equation}
holding on the KdV solution space. 
Since $a,b$ can be chosen freely,
this global conservation law can be expressed as 
$u_t =h(t) - \tfrac{1}{2}u_x^2 -u_{xxx}$,
with $h(t)$ describing an $x$-independent source/sink term
as explained in section~\ref{sec:topologicalcharge}. 
In particular, since $u$ describes a mass density, 
$h(t)>0$ will drive an increase in mass and $h(t)<0$ will drive a decrease in mass.

\subsection{KP equation}

The KP equation,
given by \cite{KadPet}
\begin{equation}\label{KP}
u_{tx} + (uu_{x} + u_{xxx})_x + \sigma u_{yy}=0,
\quad
\sigma=\pm 1
\end{equation}
up to a scaling of the coefficients,
is an integrable generalization of the KdV equation in two spatial dimensions.
It describes weakly-transverse shallow water waves when $\sigma=+1$,
and weakly-transverse thin film waves when $\sigma=-1$,
with transverse motion occurring with respect to the $x$ direction. 

This equation \eqref{KP} matches the form of the general spatial divergence PDE \eqref{pde.divform}
for $n=2$ with 
\begin{equation}
F^x=-(uu_x +u_{xxx}) = -(\tfrac{1}{2}u^2 + u_{xx})_x , 
\quad
F^y = -\sigma u_y = -(\sigma u)_y . 
\end{equation}
Note that both components $(F^x,F^y)$ are themselves a spatial divergence. 

Equation \eqref{KP} admits the two multipliers 
\begin{equation}\label{KP.Q1}
f(t),
\quad
y f(t) ,
\end{equation}
as expected from the discussion in section~\ref{sec:intro}.
A direct computation of all multipliers
that have the form $Q(t,x,y,u,u_t,u_x,u_y)$
and that involve an arbitrary function $f(t)$
yields two additional multipliers:
\begin{equation}\label{KP.Q2}
x f(t) -\tfrac{1}{2} \sigma y^2 f'(t),
\quad
xy f(t) -\tfrac{1}{6} \sigma y^3 f'(t) . 
\end{equation}
The existence of these multipliers is due to the second-order divergence form of the spatial terms in equation \eqref{KP}. 
These four multipliers \eqref{KP.Q1}--\eqref{KP.Q2}
correspond to all variational point symmetries
admitted by the Lagrangian form of the KP equation
as shown by the results in \Ref{AncGanRec2018}. 

The corresponding conservation laws of the KP equation,
which involve $f(t)$ and $f'(t)$,
are given by the following conserved currents $(\dens,\flux^x,\flux^y)$:
\begin{align}
& f(t)( 0, u_{t} +uu_{x} + u_{xxx}, \sigma u_{y} ) ;
\label{KP.conslaw1}
\\
& f(t)( 0, y(u_{t} +uu_{x} + u_{xxx}), \sigma (yu_{y}-u) ) ;
\label{KP.conslaw2}
\\
&\begin{aligned}
& f(t)\big( u , \tfrac{1}{2} u^2 +u_{xx}-x(u_{t}+uu_{x}+u_{xxx}), -\sigma x u_{y} \big)
\\&\quad
+ f'(t)\big( 0,
\tfrac{1}{2} \sigma y^2( u_{t}+ u u_{x} +u_{xxx} ) ,
{-y} u + \tfrac{1}{2} y^2 u_{y} 
\big) ;
\end{aligned}
\label{KP.conslaw3}
\\
&\begin{aligned}
& f(t)\big( yu, y(\tfrac{1}{2} u^2 +u_{xx})-xy(u_{t}+uu_{x} +u_{xxx}), \sigma x (u - yu_{y}) \big)
\\&\quad
+ f'(t)\big( 0, 
\tfrac{1}{6} \sigma y^3 ( u_{t}+ u u_{x} + u_{xxx}), 
{-\tfrac{1}{2} y^2} u +\tfrac{1}{6} y^3 u_{y} 
\big) .
\end{aligned}
\label{KP.conslaw4}
\end{align}

The first two conserved currents \eqref{KP.conslaw1} and \eqref{KP.conslaw2}
represent spatial-flux conservation laws.
In global form \eqref{2D.charge.conslaw},
they yield the topological charges 
$\oint_{C} F^y\,dx + (u_t - F^x)\,dy=0$
and 
$\oint_{C} (yF^y-\sigma u)\,dx + y(u_t - F^x)\,dy=0$
%\begin{align}
%& \oint_{C} \sigma u_y \,dx - (u_t + uu_x + u_{xxx})\,dy = 0 , 
%\\
%& \oint_{C} \sigma (yu_y -u)\,dx - y(u_t + uu_x + u_{xxx})\,dy = 0 , 
%\end{align}
holding for KP solutions $u(t,x,y)$,
where $C$ is any closed curve in the $(x,y)$-plane.
Alternatively, they can be expressed as balance equations
\begin{align}
\frac{d}{dt} \oint_{C} u \,dy & = \oint_{C} \sigma u_y \,dx - (uu_x + u_{xxx})\,dy ,
\label{KP.charge1}
\\
\frac{d}{dt} \oint_{C} y u\,dy & = \oint_{C} \sigma (yu_y -u)\,dx - y(uu_x + u_{xxx})\,dy ,
\label{KP.charge2}
\end{align}
for the circulation of the transverse mass transport vector $(0,u)$ and its $y$-moment $(0,y u)$
around closed curves $C$. 
Note these vectors lie in the $y$ direction. 

If we take $f(t)=1$,
then the conservation laws given by the other two conserved currents \eqref{KP.conslaw3} and \eqref{KP.conslaw4}
describe balance equations for rate of change of
the mass $\int_{{\rm int}(C)} u\,dxdy$ 
and for rate of change of
the $y$-moment of the mass $\int_{{\rm int}(C)} yu\,dxdy$, 
holding for KP solutions $u(t,x,y)$,
where ${\rm int}(C)$ is the interior of the closed curve $C$. 
The divergence relation \eqref{dens.div.id} for each conservation law
yields the identities 
\begin{align}
& u   =
( \tfrac{1}{2} \sigma y^2(u_{t}+u u_{x} +u_{xxx}) )_x
- (yu-\tfrac{1}{2}y^2u_{y})_y
+ \tfrac{1}{2} \sigma y^2 G , 
\label{KP.id1}
\\
& yu =
( \tfrac{1}{6} \sigma y^3(u_{t}+ u u_{x} + u_{xxx}) )_x
- (\tfrac{1}{2}y^2u -\tfrac{1}{6}y^3u_{y})_y
+ \tfrac{1}{6} \sigma y^3 G ,
\label{KP.id2}
\end{align}
where $G =   u_{tx} + (uu_{x} + u_{xxx})_x + \sigma u_{yy}$.
When the Cauchy problem for the KP equation in the form 
$u_{t} + (uu_{x} + u_{xxx}) + \sigma \partial_x^{-1} u_{yy}=0$
is considered on $\Rnum^2$,
these two identities \eqref{KP.id1} and \eqref{KP.id2} imply that
initial data $u|_{t=0}=u_0(x,y)$ must satisfy the integral constraints \cite{MolSauTzv}
$\int_{\Rnum^2} u_0\,dxdy = 0$ and $\int_{\Rnum^2} yu_0\,dxdy =0$.  
The origin of these integral constraints from conservation laws has not been previously noticed. 

Moreover, due to the identities \eqref{KP.id1} and \eqref{KP.id2}, 
the conserved currents \eqref{KP.conslaw3} and \eqref{KP.conslaw4}
are locally equivalent to the following respective spatial-flux currents
\begin{align}
&\begin{aligned}
\big( 0,&\,
\tfrac{1}{2} u^2 +u_{xx} -x(u_t + uu_{x}+u_{xxx})  -\tfrac{1}{2} \sigma y^2(u_t u_x + u_{tx} + u_{xx} + u_{txxx}), \\&\quad
yu_t -\sigma x u_{y} -\tfrac{1}{2}y^2 u_{ty}
\big) , 
\end{aligned}
\label{KP.current3}
\\
&\begin{aligned}   
\big( 0,&\,
y(\tfrac{1}{2} u^2 +u_{xx}) -xy(u_{t}+uu_{x}+u_{xxx}) -\tfrac{1}{6} \sigma y^3(u_t u_x + u_{tx} + u_{xx} + u_{txxx}), \\&\quad
\sigma x (u-yu_{y}) +\tfrac{1}{2}y^2 u_{t} -\tfrac{1}{6}y^3 u_{ty}
\big) , 
\end{aligned}
\label{KP.current4}
\end{align}
given by Theorem~\ref{thm:main}, 
where, without loss of generality, we have dropped an overall factor $f(t)$.
Their global form \eqref{2D.charge.conslaw}
yields two non-trivial topological charges
\begin{align}
&\begin{aligned}
\oint_{C} & (-\tfrac{1}{2} u^2 - u_{xx} + x(u_t + uu_x + u_{xxx} ) + \tfrac{1}{2} \sigma y^2 (u_t u_x + u_{tx} + u_{xx} + u_{txxx} ))\, dy
\\&\qquad
+ (yu_t -\sigma xu_y -2 y u_{ty})\, dx = 0 , 
 \end{aligned}
\label{KP.charge3}
\\
&\begin{aligned}
\oint_{C} & (-y(\tfrac{1}{2} u^2 +u_{xx}) +xy(uu_{x}+u_{t}+u_{xxx}) +\tfrac{1}{6} \sigma y^3(u_t u_x + u_{tx} + u_{xx} + u_{txxx}) )\,dy
\\&\qquad
+ (\sigma x (u-yu_{y}) +\tfrac{1}{2}y^2 u_{t} -\tfrac{1}{6}y^3 u_{ty} )\,dx =0 , 
\end{aligned}
\label{KP.charge4}
\end{align}
which are independent of the closed curve $C$,
for KP solutions $u(t,x,y)$. 

Finally, 
the spatial potential systems constructed from 
the four spatial-flux conserved currents \eqref{KP.conslaw1}, \eqref{KP.conslaw2}, \eqref{KP.current3}, \eqref{KP.current4}
consist of pairs of equations for $(u,w)$:
\begin{align}
&
u_t + uu_x + u_{xxx} = w_y ,
\quad
u_y =-\tfrac{1}{\sigma} w_x ;
\label{KP.potsys1}
\\
&
y(u_t + uu_x + u_{xxx}) = w_y ,
\quad
yu_y-u =-\tfrac{1}{\sigma} w_x;
\label{KP.potsys2}
\\
&
\begin{aligned} 
& \tfrac{1}{2}u^2 +u_{xx} - x(u_t + uu_x + u_{xxx} ) -\tfrac{1}{2} \sigma y^2 (u_t u_x + u_{tx} + u_{xx} + u_{txxx} ) = w_y, 
\\&
yu_t -\sigma xu_y -\tfrac{1}{2} y^2 u_{ty} = -w_x ;
\end{aligned}
\label{KP.potsys3}
\\
&
\begin{aligned}
& y(\tfrac{1}{2} u^2 +u_{xx}) -xy(u_{t}+uu_{x}+u_{xxx}) -\tfrac{1}{6} \sigma y^3(u_t u_x + u_{tx} + u_{xx} + u_{txxx}) = w_y,
\\&
\sigma x (u-yu_{y}) +\tfrac{1}{2}y^2 u_{t} -\tfrac{1}{6}y^3 u_{ty} =-w_x .
\end{aligned}
\label{KP.potsys4}
\end{align}
Each system is equivalent (modulo gauge freedom \eqref{2D.gaugefreedom})
to the KP equation.

\subsection{Universal modified KP equation}

A general modified version of the KP equation is given by
\begin{equation}\label{universalmKP}
(v_t -v^2 v_x +  \alpha v_x\partial_x^{-1}v_y + \beta vv_y+ v_{xxx})_x +\sigma v_{yy}=0
\end{equation}
up to scaling of the coefficients, where $\sigma^2= 1$.
This equation has recently been derived \cite{RatBri} as the governing equation of 
phase modulations of travelling waves in general nonlinear systems 
in 2+1 dimensions.
The case $\alpha^2=2$, $\beta=0$, $\sigma =1$ is called the mKP equation,
which is known to be integrable \cite{KonDub1984}.

In the general case,
equation \eqref{universalmKP} can be written as a PDE
in terms of a potential $u$ via $v=u_x$,
which yields
$(u_{tx} - u_x^2 u_{xx} +  \alpha u_y u_{xx} + \beta u_x u_{xy} + u_{xxxx})_x +\sigma u_{xyy}=0$.
By integrating this PDE with respect to $x$, and dropping the integration constant (a function of $t,y$), we obtain the equation 
\begin{equation}\label{mKP}
u_{tx} - u_x^2 u_{xx} +  \alpha u_y u_{xx} + \beta u_x u_{xy} + u_{xxxx} +\sigma u_{yy}=0, 
\quad
\sigma^2= 1 
\end{equation}
which we will refer to as the universal mKP (\emph{umKP}) equation.
Note that it matches the form of the general spatial divergence PDE \eqref{pde.divform}
for $n=2$ with 
\begin{equation}
F^x=\tfrac{1}{3} u_x^3 -\alpha u_y u_x -u_{xxx},
\quad
F^y =  \tfrac{1}{2} (\alpha - \beta) u_x^2 -\sigma u_y . 
\end{equation}

The umKP equation \eqref{mKP} admits the multiplier $f(t)$
given by an arbitrary function of $t$. 
Similarly to the KP equation,
another multiplier is $yf(t)$ when $F^y$ is a $y$-derivative,
namely in the case $\alpha=\beta$.

From a recent classification result in \Ref{AncGanRec2020} 
for conservation laws of a $p$-power generalization of equation \eqref{mKP}, 
we obtain four additional multipliers:
\begin{equation}
(\alpha -\beta) u_{x} f(t) + y f'(t) ;
\label{mKP.Q1}
\end{equation}
\begin{equation}
2u_y f(t) -(\tfrac{3}{4} \alpha x +\tfrac{1}{2} y u_x) f'(t) +\tfrac{3}{8} \alpha y^2 f''(t)
\label{mKP.Q2}
\end{equation}
in the case $\alpha^2=\tfrac{2}{3}$, $\beta=2\alpha$, $\sigma=1$; 
%$\alpha=\tfrac{1}{2}\beta$, $\beta^2=\tfrac{8}{3}$, $\sigma=1$
and 
\begin{gather}
(x - \alpha yu_{x}) f(t)-\tfrac{1}{2} y^2 f'(t) ,
\label{mKP.Q3}\\
(\tfrac{4}{3} \alpha u_{x} u_{y} +  \tfrac{2}{3} u_{t} -\tfrac{8}{9} u_{x}^3 + \tfrac{8}{3} u_{xxx}) )f(t)
-\tfrac{2}{3}\alpha  x u_{x} f'(t)
+\tfrac{1}{3}( y^2 -\alpha x y ) f''(t)
+\tfrac{1}{18} y^3 f'''(t) ,
\label{mKP.Q4}
\end{gather}
in the integrable case $\alpha^2=2$, $\beta=0$, $\sigma=1$. 

Each multiplier yields a conservation law of the umKP equation,
involving $f(t)$ and its derivatives. 
We will omit the lengthy expressions for the fluxes
and present only the conserved densities:
\begin{gather}
\tfrac{1}{2} (\alpha-\beta) u_{x}^2 f(t) ;
\label{mKP.dens1}
\\
u_{x} u_{y} f(t) +\tfrac{1}{4} (3\alpha u -y u_{x}^2) f'(t) ; 
\label{mKP.dens2}
\\
(u +\tfrac{1}{2} \alpha y u_{x}^2)f(t)  ; 
\label{mKP.dens3}
\\
u_{xx}^2 +\tfrac{1}{6} (u_{x}^2 -\alpha  u_{y})^2 ) f(t)
+\tfrac{1}{3} x u_{x}^2 f'(t) 
-\tfrac{1}{6}( 2\alpha y u + y^2 u_{x}^2 ) f''(t) . 
\label{mKP.dens4}
\end{gather}
If we take $f(t)=1$,
then the first density \eqref{mKP.dens1} yields the $L^2$ norm of $u_x$: 
$\| u_x \|_{L^2(\Omega)}^2 = \int_\Omega u_x^2\,dx\,dy$,
on any spatial domain $\Omega\subseteq\Rnum^2$. 
The same conserved density is admitted by the KP equation
and physically describes an $x$-momentum quantity
where $u_x=v$ is the wave amplitude.
Likewise for $f(t)=1$,
the second density \eqref{mKP.dens2} yields a $y$-momentum quantity
while the third density \eqref{mKP.dens3} yields a conserved quantity
related to the $y$-moment of the $x$-momentum.
The fourth density \eqref{mKP.dens3} for $f(t)=1$ yields the conserved energy
$E[u] = \int_\Omega ( u_{xx}^2 +\tfrac{1}{6} (u_{x}^2 -\alpha  u_{y})^2 )\,dx\,dy$,
which was first found in \Ref{KenMar} 
in a study of the Cauchy problem for the mKP equation. 

The divergence relation \eqref{dens.div.id} can be applied to each of these four conservation laws,
yielding the following remarkable identities
where $G=u_{tx} - u_x^2 u_{xx} +  \alpha u_y u_{xx} + u_{xxxx} +\sigma u_{yy}$:
\begin{align}
& u_{x}^2 = 
\tfrac{2}{\alpha-\beta}\big( D_x X +D_y Y + y G \big),
\label{mKP.id1}\\
&
X=-y( u_{t} -\tfrac{1}{3} u_{x}^3 +\alpha u_{y}u_{x} +u_{xxx} ) ,
\quad
Y = \sigma u -y( \sigma u_{y} -\tfrac{1}{2}(\alpha-\beta) u_{x}^2 ) ;
\label{mKP.id1.flux}
\end{align}
and
\begin{align}
& u_{x}u_{y} =   D_x X + D_y Y
- (\tfrac{1}{2} u_{x} y +\tfrac{3}{4} \alpha x) G
-\tfrac{1}{4\alpha} y^2 G_t , 
\label{mKP.id2}\\
&\begin{aligned}
X= & 
{-\tfrac{3}{4}} \alpha u_{xx}
+x (
\tfrac{1}{2} u_{x} u_{y}
-\alpha ( \tfrac{1}{4} u_{x}^3 -\tfrac{3}{4}(u_{t}+u_{xxx}) )
) 
\\&\quad
-y (
\tfrac{1}{8} u_{x}^4
+\tfrac{1}{4} u_{y}^2
+\tfrac{1}{4} u_{xx}^2
-\tfrac{1}{2} u_{x} u_{xxx}
-\tfrac{1}{4} \alpha u_{x}^2 u_{y}
) 
\\&\quad
+y^2 (
\tfrac{1}{4}(u_{y} u_{tx}+u_{x} u_{ty})
+\alpha \tfrac{3}{8}(u_{tt} -u_{x}^2 u_{tx} +u_{txxx})
) 
\big)
, 
\\
Y = &
x (
\tfrac{3}{4} \alpha u_{y}+\tfrac{1}{4} u_{x}^2
) 
+ y (
-\alpha(\tfrac{3}{4} u_{t} - \tfrac{1}{4} u_{x}^3) 
+\tfrac{1}{2} u_{x} u_{y}
)
+y^2 (
\tfrac{1}{4} u_{x} u_{tx}+\tfrac{3}{8} \alpha u_{ty}
) 
, 
\end{aligned}
\label{mKP.id2.flux}
\end{align}
in the case $\alpha^2=\tfrac{2}{3}$, $\beta=2\alpha$, $\sigma=1$;
and also 
\begin{align}
& y u_{x}^2 + \alpha u =  D_x X + D_y Y -\tfrac{1}{2}\alpha y^2 G , 
\label{mKP.id3}\\
& X = 
-y^2 (u_{y} u_{x} + \tfrac{1}{2}\alpha (u_{t} +\tfrac{1}{3} u_{x}^3 -u_{xxx}) 
) 
,
\quad
Y = 
\alpha y u - \tfrac{1}{2} y^2(\alpha u_{y} -u_{x}^2) 
, 
\label{mKP.id3.flux}
\end{align}
\begin{flalign}
& \begin{aligned}
& u_{xx}^2 +\tfrac{1}{6} (\alpha  u_{y} -u_{x}^2)^2
\\& =
D_x X + D_y Y
+\tfrac{1}{3} \alpha (2 x u_{x} + y^2 u_{tx}) G
-\tfrac{1}{3}(2\alpha x y - y^2 u_{x}) G_{t}
-\tfrac{1}{9}\alpha y^3 G_{tt} ,
\end{aligned}
\label{mKP.id4}\\
& \begin{aligned}
X = &
\tfrac{2}{3} u_{x} u_{xx}
-\tfrac{1}{3}\alpha y u_{txx}
+ x (
\tfrac{1}{6} u_{x}^4
+\tfrac{1}{3}  u_{y}^2
-\tfrac{1}{3}\alpha u_{x}^2 u_{y}
+\tfrac{1}{3}  u_{xx}^2
-\tfrac{2}{3} u_{x} u_{xxx}
)
\\&\quad
+ x y (
\tfrac{2}{3}( u_{x} u_{ty}+ u_{y}u_{tx} )
+\tfrac{1}{3}(u_{tt} - u_{x}^2 u_{tx}  u_{txxx})
)
\\&\quad
+y^2 (
\tfrac{1}{3}( 
u_{y}u_{ty}
+ u_{x}^3 u_{tx}
+  u_{xx} u_{txx}
- u_{tx} u_{xxx}
- u_{x} u_{txxx}
)
-\tfrac{1}{6}\alpha ( 2u_{x} u_{y} u_{tx} + u_{x}^2u_{ty} )
) 
\\&\quad
+ y^3( \tfrac{1}{9}( u_{y} u_{ttx} + u_{x} u_{tty} -2 \alpha  u_{tx} u_{ty} )
-\tfrac{1}{18}\alpha ( 2 u_{x} u_{tx}^2 + u_{x}^2 u_{ttx} -u_{ttt} -u_{ttxxx} )
)_ , 
\\
Y = & 
-x (
\alpha (\tfrac{1}{3} u_{t} -\tfrac{1}{9} u_{x}^3)
+\tfrac{2}{3} u_{x} u_{y} 
)
+\tfrac{1}{3} x y (
\alpha u_{ty}
-2 u_{x} u_{tx} 
)
\\&\quad
-\tfrac{1}{6} y^2 (
\alpha (u_{tt} -u_{x}^2 u_{tx})
+2( u_{y}u_{tx} + u_{x} u_{ty} )
) 
-\tfrac{1}{18} y^3 (
2u_{tx}^2
+2u_{x} u_{ttx}
-\alpha u_{tty}
) 
, 
\end{aligned}
\label{mKP.id4.flux}
\end{flalign}
in the integrable case $\alpha^2=2$, $\beta=0$, $\sigma=1$.

For umKP solutions $u(t,x,y)$,
the first identity \eqref{mKP.id1} implies that the $L^2$ norm of $u_x$
can be expressed as a line integral around the boundary of the spatial domain,
while the fourth identity \eqref{mKP.id4} implies that the energy norm of $u$
can be expressed in a similar form. 
Consequently,
if we consider the Cauchy problem for the umKP equation
\begin{equation}\label{mKP.ut.eqn}
u_{t} = \tfrac{1}{3}u_x^3 - u_{xxx} -\alpha u_y u_{x} -\beta u_x u_{xy} +\partial_x^{-1}(\tfrac{1}{2}(\alpha-\beta) u_x^2 -\sigma u_{y})_y
\end{equation}
on $\Rnum^2$,
then it cannot be well-posed in the following two situations:
\begin{equation}
\| u_x \|_{L^2(\Rnum^2)} <\infty
\quad\text{ and }\quad
\lim_{r\to\infty}\oint_{S^1(r)} ( Y_{L^2} \sin\theta + X_{L^2} \cos\theta )|_\Esp r\,d\theta =0 
\end{equation}
where $(X_{L^2},Y_{L^2})$ is given by expressions \eqref{mKP.id1.flux}; 
or
\begin{equation}
\alpha^2=2,
\
\beta=0,
\
\sigma=1,
\
E[u] <\infty
\ \text{ and }
\lim_{r\to\infty}\oint_{S^1(r)} ( Y_{E} \sin\theta + X_{E} \cos\theta )|_\Esp r\,d\theta =0 
\end{equation}
where $(X_{E},Y_{E})$ is given by expressions \eqref{mKP.id4.flux}.
In both line integrals,
$u_t$ is eliminated through equation \eqref{mKP.ut.eqn},
so that the integrand depends on only $u$ and its $x,y$-derivatives.

The conservation laws corresponding to the four conserved densities \eqref{mKP.dens1}--\eqref{mKP.dens4}
are each locally equivalent to a spatial flux conservation law by Theorem~\ref{thm:main},
which yields four conserved non-trivial topological charges for the umKP equation.
In particular, through the divergence relation \eqref{mKP.id3},
the energy quantity can be expressed as a topological charge:
\begin{equation}
E[u] = \oint_{\partial\Omega} (X_{E}\, dy - Y_{E}\, dx )|_\Esp 
%&
%\big( ( \alpha ( \tfrac{1}{9} u_{x}^3-\tfrac{1}{3} u_{t} ) -\tfrac{2}{3} u_{x} u_{y}) x 
%+( \tfrac{1}{3} a u_{ty} -\tfrac{2}{3} u_{x} u_{tx} ) xy
% +( \tfrac{1}{6}\alpha (u_{tx} u_{x}^2-u_{tt}) -\tfrac{1}{3}( u_{tx} u_{y} + u_{x} u_{ty} ) ) y^2
% -( \tfrac{1}{9} u_{tx}^2 +\tfrac{1}{9} u_{x} u_{ttx} -\tfrac{1}{18} \alpha u_{tty} ) y^3 
% \big)dx
%-\big( \tfrac{2}{3} u_{x} u_{xx} -\tfrac{1}{3} \alpha u_{txx}\, y 
% +( \tfrac{1}{3} u_{xx}^2 +\tfrac{1}{6} u_{x}^4 +\tfrac{1}{3} u_{y}^2 -\tfrac{2}{3} u_{x} u_{xxx} -\tfrac{1}{3} \alpha u_{y} u_{x}^2 ) x
% +( \tfrac{1}{3} \alpha (u_{txxx} - u_{tx} u_{x}^2 + u_{tt}) +\tfrac{2}{3} ( u_{tx} u_{y} + u_{x} u_{ty} ) ) xy
% +( \tfrac{1}{3} ( u_{xx} u_{txx} + u_{tx} u_{x}^3 + u_{ty} u_{y} -u_{x} u_{txxx} - u_{tx} u_{xxx} ) -\alpha (\tfrac{1}{3} u_{x} u_{y} u_{tx} +\tfrac{1}{6} u_{x}^2 u_{ty}) ) y^2
% +( \tfrac{1}{18}\alpha ( u_{ttxxx} + u_{ttt} - u_{ttx} u_{x}^2 -2 u_{x} u_{tx}^2 ) +\tfrac{1}{9}( 2 u_{tx} u_{ty} + u_{x} u_{tty} + u_{y} u_{ttx} ) ) y^3
% \big)dy
\end{equation}
where $\partial\Omega$ is the boundary of the spatial domain $\Omega$. 
This remarkable integral relation provides
an explanation to a question raised in \Ref{KenMar}
about why the conserved energy could not be used to obtain 
a global existence result for solutions to the Cauchy problem for the integrable mKP equation.

Finally, each topological charge 
%spatial-flux conservation law 
gives rise to an associated spatial potential system for the umKP equation, 
in the same fashion as for the KP equation.

\subsection{Shear-wave equations}

The propagation of shear waves in incompressible nonlinear solids
can be modelled by the equation
\begin{equation}\label{ZZKeqn}
(u_{t} + (\alpha u +\beta u^2)u_{x})_x + \Delta_\perp u=0
\end{equation}
up to a scaling of the coefficients.
Here $\Delta_\perp = \partial_y^2 +\partial_z^2$ is the transverse Laplacian. 
For $\alpha=0$,
equation \eqref{ZZKeqn} describes nonlinear, linearly-polarized shear waves \cite{Zab},
%Zabolotskaya (Z) equation 
while for $\beta=0$,
it describes weakly nonlinear, weakly diffracting shear waves \cite{DesLaiOriSac,ZabKho}. 
In the latter case, the equation is known as the Khokhlov--Zabolotskaya equation,
and its local conservation laws have been found in \Ref{Sha}. 

This equation has the form of a dispersionless Gardner equation extended to three spatial dimensions.
It matches the form of the general spatial divergence PDE \eqref{pde.divform}
for $n=3$ with 
\begin{equation}
F^x=-(\alpha u +\beta u^2)u_{x} = -(\tfrac{1}{2}\alpha u^2 +\tfrac{1}{3}\beta u^3)_x , 
\quad
F^y = -u_y ,
\quad
F^z = -u_z ,
\end{equation}
where all components $(F^x,F^y,F^z)$ are themselves a spatial divergence. 

Equation \eqref{ZZKeqn} admits the multipliers $f(t)$, $yf(t)$, $zf(t)$,
in accordance with the discussion in section~\ref{sec:intro}.
In fact, these multipliers are special cases of a more general multiplier
\begin{equation}\label{ZZK.Q}
f_t(t,y,z) -x \Delta_\perp f(t,y,z),
\quad
\Delta_\perp^2 f(t,y,z) =0 , 
\end{equation}
which is obtained by a direct computation of all multipliers 
that have the form $Q(t,x,y,z,u,u_t,u_x,u_y,u_z)$
with arbitrary dependence on $t$.
Note that if $f(t,z,y)$ is assumed to be analytic in $t$
then without loss of generality 
we can put $f(t,y,z)=\tilde f(t) \phi(y,z)$
where $\phi$ satisfies the biharmonic equation $\Delta^2\phi=0$,
and where $\tilde f(t)$ is an arbitrary function of $t$. 

The conservation law arising from this multiplier \eqref{ZZK.Q}
is given by the following conserved current $(\dens,\flux^x,\flux^y,\flux^z)$:
\begin{equation}\label{ZZK.current}
\begin{aligned}
\big( & 
(\Delta_\perp u) f , 
-(x \Delta_\perp u_t) f
+( \tfrac{1}{2}\alpha u^2 +\tfrac{1}{3}\beta u^3 -x(\alpha u +\beta u^2) u_x )\Delta_\perp f
+(u_t + (\alpha u +\beta u^2)u_x)f_t 
,
\\&
x( u_{txy} f -u_{tx} f_y -u_y \Delta_\perp f + u \Delta_\perp f_y ), 
x( u_{txz} f -u_{tx} f_z -u_z \Delta_\perp f + u \Delta_\perp f_z )
\big) . 
\end{aligned}
\end{equation}
For $f=\phi$,
this conserved current describes a balance equation
for the rate of change of  the weighted mass 
$\int_{{\rm int}(S)} \phi \Delta_\perp u \,dxdydz$, 
holding for solutions $u(t,x,y,z)$ of equation \eqref{ZZKeqn},
where ${\rm int}(S)$ is the interior of any volume bounded by a closed surface $S$ 
in $(x,y,z)$-space. 

The divergence relation \eqref{dens.div.id} applied to the conserved current \eqref{ZZK.current}
yields the identity
\begin{equation}\label{ZZK.id}
\phi \Delta_\perp u  = D_x\big( (u_t + (\alpha u +\beta u^2)u_x)\phi \big) +\phi G ,
\end{equation}
where $G= (u_{t} + (\alpha u +\beta u^2)u_{x})_x + \Delta_\perp u$.
Consequently, the conserved current is locally equivalent to
a spatial flux conservation law.
The resulting topological charge is given by 
\begin{equation}\label{ZZK.charge}
\oint_{S} (u_t+ (\alpha u +\beta u^2)u_x)\phi \,dydz =0 .
\end{equation}
It can be expressed as a balance equation 
\begin{equation}
\frac{d}{dt}\oint_{S} u \phi \,dydz = -\oint_{S} (\alpha u +\beta u^2)u_x \phi \,dydz
\end{equation}
for the weighted flux of the mass transport vector $(u,0,0)$ 
through closed surfaces $S$,
with weighting factor $\phi$. 
In particular, if $S$ is taken to comprise the planar surfaces bounding a rectangular volume with infinite extent in the $y$ and $z$ directions, 
then $\oint_{S} u \,dydz$ is the net flux through the two transverse surfaces, 
say $x=a$ and $x=b$, for solutions $u(t,x,y)$ with sufficient transverse decay. 

The spatial potential system arising from the topological charge \eqref{ZZK.current}
%the spatial-flux conserved current \eqref{ZZK.current}
consists of the following system of equations for $(u,w^x,w^y,w^z)$:
\begin{equation}
\begin{aligned}
& ( \tfrac{1}{2}\alpha u^2 +\tfrac{1}{3}\beta u^3 -x(\alpha u +\beta u^2) u_x )\Delta_\perp\phi
  -( u_{tt} + (\alpha +2\beta u)u_{t}u_{x} + (\alpha u +\beta u^2) u_{tx} +x\Delta_\perp u_t )\phi
\\&  
=(w^z)_y - (w^y)_z,
\\&
x( u_{txy} \phi -u_y \Delta_\perp\phi -u_{tx} \phi_y + u \Delta_\perp\phi_y ) 
= (w^x)_z - (w^z)_x,
\\&
x( u_{txz} \phi -u_{tx} \phi_z -u_z \Delta_\perp\phi + u \Delta_\perp \phi_z ) = (w^y)_x - (w^x)_y,
\end{aligned}
\end{equation}
where $\phi$ is an arbitrary biharmonic function. 
This system is equivalent (modulo gauge freedom \eqref{3D.gaugefreedom})
to equation \eqref{ZZKeqn}.

Finally,
when the Cauchy problem is considered for equation \eqref{ZZKeqn} in the form
$u_{t} + (\alpha u +\beta u^2)u_{x} + \partial_x^{-1} \Delta_\perp u=0$
on $\Rnum^3$,
the identity \eqref{ZZK.id} implies that
initial data $u|_{t=0}=u_0(x,y,z)$ must satisfy the integral constraint 
$\int_{\Rnum^3} \phi \Delta_\perp u_0 \,dxdydz =0$
for existence of solutions $u(t,x,y,z)$ with sufficiently strong spatial decay. 
Similarly, 
if Dirichlet boundary conditions on $u$ are imposed at $x=a$ and $x=b$, 
then the identity \eqref{ZZK.id} imposes an integral constraint 
$\int_{V} \phi \Delta_\perp u\,dx\,dy\,dz=0$ 
on solutions of the corresponding boundary-value problem 
posed in infinite rectangular volume $V$ enclosed by the surfaces $x=a$ and $x=b$.

%\subsection{Jimbo-Miwa equation}
%The Jimbo-Miwa (JM) equation
%\begin{equation}
%u_{yt} + u_x u_{xy} + u_y u_{xx} + u_{xz} + u_{xxxy} =0
%\end{equation}  
%is of interest in the theory of integrable systems.
%It belongs to the KP hierarchy which is an infinite set of nonlinear PDEs
%starting with the KP equation and having a bilinear formulation in terms of a tau function.
%But it is not itself an integrable system.

\subsection{Novikov--Veselov equation}

The Novikov--Veselov (NV) equation \cite{VesNov,NovVes}
\begin{equation}\label{NV.isoflow}
v_t +\alpha (v\partial_y^{-1}v_x)_x +\beta (v\partial_x^{-1}v_y)_y + v_{xxx} + v_{yyy} =0
\end{equation}
generates isospectral flows for the two-dimensional Schr\"odinger operator
at zero energy. 
It generalizes the KdV equation in the sense that 
any NV solution $v(t,x,y)$ produces a KdV solution given by $v(t,x,x)$
up to a scaling of $t,x,v$.

Equation \eqref{NV.isoflow} can be written as a PDE
by introducing a potential $u$
given by $v=u_{xy}$, which satisfies
\begin{equation}\label{NVeqn}
u_{txy} +\alpha (u_{xy}u_{xx})_x +\beta (u_{xy} u_{yy})_y + u_{xxxxy} + u_{xyyyy} =0 . 
\end{equation}
Special cases of this PDE are $\alpha=0$ and $\beta=0$.
In these cases, the PDE can be expressed in a lower-order form
through $\tilde u=u_y$ or $\tilde u=u_x$, which respectively yield
%\begin{subequations}
\begin{align}
\tilde u_{tx} + \beta (\tilde u_{x} \tilde u_{y})_y + \tilde u_{xxxx} + \tilde u_{xyyy} =0,
%\label{aNVeqn}
%\\
\quad
\tilde u_{ty} + \alpha (\tilde u_{x} \tilde u_{y})_x + \tilde u_{xxxy} + \tilde u_{yyyy} =0 . 
\end{align}
%\end{subequations}
These two PDEs are related by $x\leftrightarrow y$ and $\alpha\leftrightarrow \beta$;
they are sometimes referred to as the asymmetric NV equation. 

The NV equation \eqref{NVeqn} has a higher-order spatial divergence form
$u_{txy} =F_{xy}+ (F^x)_x + (F^y)_y$ given by 
$F= -(u_{xxx} + u_{yyy})$, 
$F^x=  - \alpha u_{xy}u_{xx}$, 
$F^y =  -\beta u_{xy} u_{yy}$. 
This equation (with $\alpha,\beta\neq0$) admits the multipliers
\begin{equation}\label{NV.Qs}
f(t),
\quad
2 u_{xx} f(t)- (1/\alpha) x f'(t) ,
\quad
2 u_{yy} f(t)- (1/\beta) y f'(t)  , 
\end{equation}  
which come from a direct computation of all multipliers
that have the form $Q(t,x,y,u,u_t,u_x,u_y,u_{tx},u_{ty},u_{xy},u_{xx},u_{yy})$. 

The conservation law corresponding to the first multiplier
is given by the following conserved current $(\dens,\flux^x,\flux^y)$:
\begin{equation}\label{NV.current1}
f(t)\big( 0, 
u_{xxxy}+\alpha u_{xy} u_{xx} +\tfrac{1}{2} u_{ty}, 
u_{xyyy}+\beta u_{xy} u_{yy}+\tfrac{1}{2} u_{tx}
\big) . 
\end{equation}
Likewise,
the conservation laws corresponding to the second and third multipliers
are given by a pair of conserved currents
\begin{equation}\label{NV.current2}
\begin{aligned}
& f(t)\big( u_{xx} u_{xy} , 
-u_{tx} u_{xy} \alpha u_{xx}^2 u_{xy} -\beta u_{xy}^2 u_{yy} +u_{xyy}^2 +2 u_{xx} u_{xxxy} ,
\\&\quad
\tfrac{1}{3} (\alpha u_{xx}^3 + \beta u_{xy}^3)
-u_{xxx}^2 -2 u_{xxy} u_{xyy} 
+(u_{tx}+2 \beta u_{xy} u_{yy}+2 u_{xyyy}) u_{xx}
\big)
\\&
+f'(t) \big( 0,
(1/\alpha) u_{xxx} 
-x (u_{xx} u_{xy} +(1/\alpha) u_{xxxy}) ,
-(1/\alpha) x (u_{tx} +\beta u_{xy} u_{yy} +u_{xyyy}) 
\big) 
\end{aligned}
\end{equation}
and 
\begin{equation}\label{NV.current3}
\begin{aligned}
& f(t)\big( u_{yy} u_{xy} ,
\tfrac{1}{3} (\alpha u_{xy}^3 + \beta u_{yy}^3)
-u_{yyy}^2 -2 u_{xxy} u_{xyy}
+(u_{ty} +2\alpha u_{xx} u_{xy} +2  u_{xxxy})u_{yy} ,
\\&\quad
-\alpha u_{xx} u_{xy}^2
+\beta u_{xy} u_{yy}^2
-u_{ty} u_{xy}
+u_{xxy}^2
+2 u_{yy} u_{xyyy}
\big)
\\&
+ f'(t) \big( 0,
-(1/\beta) y (\alpha u_{xy} u_{xx} +u_{ty}+u_{xxxy}) , 
(1/\beta) u_{yyy} -y (u_{xy} u_{yy} +(1/\beta) u_{xyyy} ) 
\big) 
\end{aligned}
\end{equation}
which are related by reflection under $x\leftrightarrow y$. 

The first conserved current \eqref{NV.current1} 
represents a spatial-flux conservation law.
Its global form \eqref{2D.charge.conslaw} is a conserved topological charge
$\oint_{C} Y\,dx - X\,dy=0$ 
given by 
\begin{equation}
X = \tfrac{1}{2} u_{ty} +\alpha u_{xx}u_{xy} + u_{xxxy},
\quad
Y = \tfrac{1}{2} u_{tx} +\beta u_{xy}u_{yy} + u_{xyyy}
\label{NV.charge1}
\end{equation}
for NV solutions $u(t,x,y)$,
where $C$ is any closed curve in the $(x,y)$-plane.
This conservation law can also be expressed as a balance equation
\begin{equation}
\frac{d}{dt} \oint_{C} u_{x}\,dx - u_{y}\,dy = 
2\oint_{C} (\alpha u_{xx}u_{xy} + u_{xxxy})\,dy -(\beta u_{xy}u_{yy} + u_{xyyy})\,dx 
\end{equation}
for the circulation of the vector $(-u_x,u_y)$ whose curl is $2u_{xy}=2v$. 

If we take $f(t)=1$,
then the other two conserved currents \eqref{NV.current2} and \eqref{NV.current3} 
represent conservation laws describing balance equations
for rate of change of momenta
$\int_{{\rm int}(C)} u_{xx} u_{xy}\,dxdy$
and $\int_{{\rm int}(C)} u_{xy} u_{yy}\,dxdy$, 
holding for NV solutions $u(t,x,y)$,
where ${\rm int}(C)$ is the interior of any closed curve $C$ 
in the $(x,y)$-plane.
The divergence relation \eqref{dens.div.id} applied to these conservation laws
yields the respective identities 
\begin{align}
& u_{xx} u_{xy} 
= (1/\alpha)\big( D_x( x (\alpha u_{xx} u_{xy} + u_{xxxy}) )
+ D_y( x (u_{tx} +\beta u_{xy} u_{yy} + u_{xyyy}) -u_{xxx} )
- x G f(t) \big),
\label{ZZK.id1}
\\
& u_{xy} u_{yy} 
= (1/\beta)\big( D_x( y (u_{ty} +\alpha u_{xx} u_{xy} + u_{xxxy}) -u_{yyy} ) 
+ D_y( y (\beta u_{xy} u_{yy} + u_{xyyy}) )
-y G f(t) \big),
\label{ZZK.id2}
\end{align}
where $G = u_{txy} +\alpha (u_{xy}u_{xx})_x +\beta (u_{xy} u_{yy})_y + u_{xxxxy} + u_{xyyyy}$. 
When the Cauchy problem for the NV equation \eqref{NVeqn} is considered on $\Rnum^2$,
these identities \eqref{ZZK.id1} and \eqref{ZZK.id2} imply that
initial data $v|_{t=0}=v_0(x,y)=u_0{}_{xy}$
must satisfy the integral constraints 
$\int_{\Rnum^2} u_0{}_{xx}u_0{}_{xy}\,dxdy = \int_{\Rnum^2} u_0{}_{yy}u_0{}_{xy}\,dxdy = 0$. 

In addition, the identities \eqref{KP.id1} and \eqref{KP.id2}
show that each of the momenta conservation laws is locally equivalent to
a spatial flux conservation law.
The corresponding topological charges
$\oint_{C} Y\,dx - X\,dy=0$ for any closed curve $C$ 
are given by the fluxes
\begin{equation}\label{NV.charge2}
\begin{aligned}
X = & 
x (u_{xy} u_{txx}+u_{xx} u_{txy}+ (1/\alpha) u_{txxxy}) 
-u_{xy} u_{tx} +\alpha u_{xx}^2 u_{xy} -\beta u_{xy}^2 u_{yy} 
+u_{xyy}^2 +2 u_{xx} u_{xxxy} , 
\\
Y = &
(1/\alpha) x ( u_{ttx} +\beta (u_{yy} u_{txy} + u_{xy} u_{tyy}) +u_{txyyy} ) 
+\tfrac{1}{3}(\alpha u_{xx}^3 +\beta u_{xy}^3)
-u_{xxx}^2
\\&\quad
+(u_{tx} + 2 \beta u_{xy} u_{yy}  +2 u_{xyyy}) u_{xx}
-2  u_{xxy} u_{xyy} -(1/\alpha) u_{txxx} , 
\end{aligned}
\end{equation}
and
\begin{equation}\label{NV.charge3}
\begin{aligned}
X = &
(1/\beta) x ( u_{tty} +\alpha (u_{xx} u_{txy} + u_{xy} u_{txx}) +u_{txxxy} ) 
+\tfrac{1}{3}(\alpha u_{xx}^3 +\beta u_{xy}^3)
-u_{yyy}^2
\\&\quad
+(u_{ty} +2 \alpha u_{xy} u_{xx} +2 u_{xxxy}) u_{yy}
-2  u_{xxy} u_{xyy} -(1/\beta) u_{tyyy} ,
\\
Y = & 
x (u_{xy} u_{tyy}+u_{yy} u_{txy}+ (1/\beta) u_{txyyy}) 
-u_{xy} u_{ty} -\alpha u_{xy}^2 u_{xx} +\beta u_{yy}^2 u_{xy} 
+u_{xxy}^2 +2 u_{yy} u_{xyyy} . 
\end{aligned}
\end{equation}

Finally, 
the spatial potential systems arising from the three topological charges \eqref{NV.charge1}, \eqref{NV.charge2} and \eqref{NV.charge3}
%spatial-flux conservation laws 
consist of pairs of equations for $(u,w)$:
\begin{equation}
X= w_y, 
\quad
Y = -w_x . 
\end{equation}
Each system is equivalent (modulo gauge freedom \eqref{2D.gaugefreedom})
to the NV equation.

\subsection{Vorticity equation}

In two spatial dimensions,
the Navier--Stokes equations for incompressible fluid flow
\begin{equation}
\vec{v}_t + \vec{v}\cdot\nabla\vec{v}
= -(1/\rho)\nabla P + \mu \Delta\vec{v},
\quad
\nabla\cdot\vec{v}=0
\end{equation}
have a well-known formulation in terms of a scalar potential $u$,
where $\mu$ is the viscosity,
and where $\vec{v} = (-u_y,u_x)$ is the fluid velocity.
The vorticity of the fluid is a scalar given by \cite{MajBer}
$\omega = (v^y)_x-(v^x)_y=\Delta u$,
with $\Delta = \partial_x^2 + \partial_y^2$
being the two-dimensional Laplacian.
An evolution equation for $\omega$ is obtained by taking the curl of the velocity equation,
yielding
\begin{equation}\label{vorteqn}
\Delta u_t +u_x \Delta u_y -u_y \Delta u_x -\mu \Delta^2 u =0 . 
\end{equation}
This equation has a higher-order spatial divergence form
$\Delta u_t = \nabla\cdot\vec{F}$
given by 
\begin{equation}
\vec{F} = (u_y\Delta u +\mu \Delta u_x, -u_x\Delta u +\mu \Delta u_y) \\
%= ( D_x(u_xu_y+\mu\Delta u)+\tfrac{1}{2}D_y(u_y{}^2-u_x{}^2),D_y({-}u_xu_y+\mu\Delta u)-\tfrac{1}{2}D_x(u_x{}^2-u_y{}^2)  )
= -(\Delta u)\vec v + \mu \vec\nabla \Delta u . 
\end{equation}

As expected from the discussion in section~\ref{sec:intro},
equation \eqref{vorteqn} admits the multipliers
\begin{equation}
f(t) , 
\end{equation}
and in the inviscid case $\mu=0$,
\begin{equation}
xf(t),
\quad
yf(t) . 
\end{equation}
There are no additional multipliers 
that have the form $Q(t,x,y,u,u_t,u_x,u_y,u_{tx},u_{ty},u_{xy},$ $u_{xx},u_{yy})$
and that involve an arbitrary function $f(t)$,
as shown by a direct computation. 

The corresponding conservation laws are given by
the following conserved currents $(\dens,\flux^x,\flux^y)$:
\begin{equation}
f(t)\big( 0,  u_{tx} -u_y\Delta u -\mu \Delta u_x, u_{ty} +u_x\Delta u -\mu \Delta u_y \big) ;
\end{equation}
and in the case $\mu=0$, 
\begin{align}
& f(t)\big( 0, 
-u_t +u_x u_y + x(u_{tx} -u_y u_{xx} + u_x u_{xy}) ,
-u_x^2 + x (u_{ty} -u_y u_{xy} + u_x u_{yy}) 
\big) ,
\\
& f(t)\big( 0, 
u_y^2 +y (u_{tx} -u_y u_{xx} +u_x u_{xy}) , 
-u_t -u_x u_y + y (u_{ty} -u_y u_{xy} +u_x u_{yy}) 
\big) . 
\end{align}
The resulting topological charges are given by the line integrals
\begin{align}
&  \oint_{C} (u_{ty} -F^y)\,dx -(u_{tx} -F^x)\,dy =0 , 
%\oint_{C} (u_{ty} +u_x\Delta u -\mu \Delta u_y)\,dx -(u_{tx} -u_y\Delta u -\mu \Delta u_x)\,dy =0 , 
\label{vort.charge1}\\
& \oint_{C} ( x(u_{ty} -u_y u_{xy} + u_x u_{yy}) -u_x^2 )\,dx -( x(u_{tx} -u_y u_{xx} + u_x u_{xy}) -u_t +u_x u_y  )\,dy =0 , 
\label{vort.charge2}\\
& 
\oint_{C} ( y(u_{ty} -u_y u_{xy} +u_x u_{yy}) -u_t - u_x u_y )\,dx -( y(u_{tx} -u_y u_{xx} +u_x u_{xy}) +u_y^2 )\,dy =0 , 
\label{vort.charge3}
\end{align}
where $C$ is any fixed closed curve in the $(x,y)$-plane.
They can be expressed equivalently in the form of circulation balance equations
\begin{align}
\frac{d}{dt}\oint_{C} u_{y}\,dx -u_{x}\,dy 
& = \oint_{C} F^y\,dx -F^x\,dy ,
\\
\frac{d}{dt}\oint_{C} x u_{y}\,dx +(u -xu_{x})\,dy 
& = \oint_{C} ( x(u_y u_{xy} -u_x u_{yy}) +u_x^2 )\,dx -( x(u_y u_{xx} - u_x u_{xy}) -u_x u_y  )\,dy , 
\\
\frac{d}{dt} \oint_{C} (y u_{y} -u)\,dx - y u_{x}\,dy 
& = \oint_{C} ( y(u_y u_{xy} -u_x u_{yy}) +u_x u_y )\,dx -(y(u_y u_{xx} -u_x u_{xy}) +u_y^2 )\,dy .
\end{align}
The first equation relates the rate of change 
in the circulation of fluid velocity $\oint_{C} \vec v\cdot d\vec s$ around closed curves
to the flux of $\vec F$ through the curve. 
The other two equations are generalizations involving a weighted circulation of fluid velocity. 

The spatial potential systems for $(u,w)$ arising from 
the topological charges \eqref{vort.charge1}--\eqref{vort.charge3}
are given by 
\begin{align}
&
u_{tx} -u_y\Delta u -\mu \Delta u_x =w_y,
\quad
u_{ty} +u_x\Delta u -\mu \Delta u_y =-w_x; 
\\
&
{-}u_t  +u_x u_y + x (u_{tx} - u_y u_{xx} + u_x u_{xy}) 
= w_y,
\quad
u_x^2 - x (u_{ty} - u_y u_{xy} + u_x u_{yy}) 
= w_x ,
\\
& 
u_y^2 +y (u_{tx} -u_y u_{xx} +u_x u_{xy}) 
= w_y,
\quad
u_t + u_x u_y - y (u_{ty} -u_y u_{xy} +-u_x u_{yy}) 
=w_x . 
\end{align}

\section{Concluding remarks}\label{sec:conclude}

For dynamical PDEs with a spatial divergence form, 
we have shown how conservation laws that involve an arbitrary function of time
contain interesting useful information which goes beyond what is contained
in ordinary conservation laws.

Conservation laws in one spatial dimension
involving an arbitrary function $f(t)$ 
describe the presence of an $x$-independent source/sink; 
in two and more spatial dimensions,
such conservation laws are locally equivalent to spatial-flux conservation laws
that yield non-trivial topological charges.

We have explored two important systematic applications of
this type of conservation law:
construction of an associated spatial potential system which allows
the possibility for finding nonlocal symmetries and nonlocal conservation laws; 
and derivation of an integral constraint relation which imposes
restrictions on well-posedness of the Cauchy problem
with associated initial/boundary data. 

It is well known that the existence of a non-trivial conservation law of a PDE
has a cohomological meaning in the setting of the variation bi-complex \cite{Olv-book}.
Existence of a topological charge can be understood as stating that
a given dynamical PDE possesses a non-trivial spatial cohomology. 

All of these developments will have a natural counterpart
for conservation laws that involve an arbitrary function of
other independent variables.

\section*{Acknowledgements}
SCA is supported by an NSERC research grant
and thanks the University of C\'adiz for additional support during the period 
when this work was initiated.

\end{document}